\documentclass[12pt,onecolumn,journal,draftclsnofoot,letterpaper]{IEEEtran}
\usepackage{amssymb,amsthm}
\usepackage{amsxtra,amsfonts}
\usepackage[dvipsnames]{xcolor}
\usepackage{enumitem}
\usepackage{enumitem}
\usepackage{latexsym}
\usepackage{mathrsfs}
\usepackage{verbatim}
\usepackage{algorithm}
\usepackage{algorithmicx,algcompatible}
\usepackage{algpseudocode}
\usepackage{cite}
\usepackage{graphicx}
\usepackage[pdfborder={0 0 0},colorlinks,citecolor=blue,hyperindex=true,hypertexnames=false,extension=pdf]{hyperref}
\usepackage[caption=false,font=footnotesize]{subfig}
\algdef{SE}[DOWHILE]{Do}{doWhile}{\algorithmicdo}[1]{\algorithmicwhile\ #1}%
\algnewcommand\INPUT{\item[\textbf{Input:}]}%
\algnewcommand\OUTPUT{\item[\textbf{Output:}]}%

\newtheorem{lemma}{Lemma}
\newtheorem{theorem}{Theorem}

\newcommand{\pw}[1]{{\color{black}#1}}
\newcommand{\pwtwo}[1]{{#1}}
\newcommand{\Expect}[1]{{\ensuremath{\mathbb E}[#1]}}

\newcommand{\vQ}{\mathbf P}
\newcommand{\vP}{\mathbf Q}
\newcommand{\df}{\ensuremath{\lam}}         
\newcommand{\gP}{\ensuremath{\vP}}          
\newcommand{\gp}{\ensuremath{\vp}}          
\newcommand{\fH}{\ensuremath{\vh}}          
\newcommand{\bP}{\ensuremath{\vp}}          
\newcommand{\bH}{\ensuremath{\vh}}          
\newcommand{\Pbar}{\ensuremath{\bar{P}}}         
\newcommand{\Rb}{\ensuremath{R_b}}         
\newcommand{\Vor}{\ensuremath{\mathcal{V}}}         

\newcommand{\Rset}{\ensuremath{\mathcal R}}

\newcommand{\Qset}{\ensuremath{\mathcal P}}

\newcommand{\Dis}{\ensuremath{P}}                    
\newcommand{\AvDis}{\ensuremath{\bar{P}}}         
\newcommand{\AvPtx}{\ensuremath{\bar{P}}}         
\newcommand{\Ptxn}{\ensuremath{P_{n}}}         
\newcommand{\Ptx}{\ensuremath{P_{\text{TX}}}}         
\newcommand{\Prxn}{\ensuremath{P_{\text{RX},n}}}         
\newcommand{\inPtxn}{\ensuremath{P_{\text{TX},n}}}         

\newcommand{\Grxn}{\ensuremath{G_{\text{RX},n}}}         
\newcommand{\Gtx}{\ensuremath{G_{\text{TX}}}}         

\newcommand{\GHPBW}{\ensuremath{G_{\text{HPBW}}}}         
\newcommand{\Ghpbw}{\ensuremath{G_{\text{HPBW}}}}         
\newcommand{\thtHPBW}{\ensuremath{\theta_{\text{HPBW}}}}         
\newcommand{\ththpbw}{\ensuremath{\theta_{\text{HPBW}}}}         
\newcommand{\hmin}{\ensuremath{h_{\text{min}}}}         
\newcommand{\RhminN}{\ensuremath{\R^N_{h_{\text{min}}}}}         
\newcommand{\Plos}{\ensuremath{\P_{\text{LoS}}}}         
\newcommand{\hPlos}{\ensuremath{\hat{\bet}_{\text{LoS}}}}         
\newcommand{\betnlos}{\ensuremath{\bet_{\text{NLoS}}}}         
\newcommand{\kapnlos}{\ensuremath{\kap_{\text{NLoS}}}}         
\newcommand{\thtelevation}{\ensuremath{\tht_\text{E}}}         
\newcommand{\vome}{\ensuremath{\boldsymbol{ \ome}}}
\renewcommand{\P}{\operatorname{Pr}}
\newcommand{\zero}{{\ensuremath{\mathbf 0}}}
\newcommand{\LRA}{\ensuremath{\Leftrightarrow} }

\newcommand{\N}{{\ensuremath{\mathbb N}}}
\newcommand{\Nplus}{\ensuremath{\N_+}} 
\newcommand{\R}{{\ensuremath{\mathbb R}}}

\newcommand{\Norm}[1]{\ensuremath{ \left\|#1\right\| }}
\newcommand{\alp}{\ensuremath{\alpha}}
\newcommand{\bet}{\ensuremath{\beta}}
\newcommand{\Del}{\ensuremath{\Delta}}

\newcommand{\eps}{\ensuremath{\epsilon}}

\newcommand{\gam}{\ensuremath{\gamma}}

\newcommand{\lam}{\ensuremath{\lambda}}
\newcommand{\kap}{\ensuremath{\kappa}}
\newcommand{\Ome}{\ensuremath{\Omega}}
\newcommand{\ome}{\ensuremath{\omega}}

\newcommand{\sig}{\ensuremath{\sigma}}
\newcommand{\tht}{\ensuremath{\theta}}

\newcommand{\vc}{{\ensuremath{\mathbf c}}}

\newcommand{\vh}{{\ensuremath{\mathbf h}}}                         
\newcommand{\vp}{{\ensuremath{\mathbf p}}}
\newcommand{\vq}{{\ensuremath{\mathbf q}}}

\newcommand{\vx}{{\ensuremath{\mathbf x}}}

\newcommand{\thmref}[1]{Theorem~\ref{#1}}     
\newcommand{\lemref}[1]{Lemma~\ref{#1}}       

\newcommand{\appref}[1]{Appendix~\ref{#1}}
\newcommand{\secref}[1]{Section~\ref{#1}}

\newcommand{\figref}[1]{Fig.~\ref{#1}}

\newcommand{\skprod}[2]{\ensuremath{ \left\langle #1,#2 \right\rangle }}
\newcommand{\noi}{\noindent}
\newcommand{\tM}{\ensuremath{\tilde{M}}}
\newcommand{\tc}{\ensuremath{\tilde{c}}}

\newenvironment{remark}{\par\vspace{1.5ex}\noindent{\em Remark\/}.}{\par\vspace{1.5ex}}

\newcounter{example}[section]
\newenvironment{example}[1][]{\refstepcounter{example}\par\vspace{1.5ex}\noindent{\em Example. #1}}{\par\vspace{1.5ex}}

\DeclareMathOperator{\Clos}{Clos}

\makeatletter
  \newcommand{\set}[2]{\ensuremath{%
  \setbox0=\hbox{\ensuremath{#2}}
  \dimen@\ht0
  \advance\dimen@ by \dp0
  \left\{\left.#1\rule[-\dp0]{0pt}{\dimen@}\;\right|\;#2\right\} }}
\makeatother

\makeatletter
\def\blfootnote{\xdef\@thefnmark{}\@footnotetext}
\makeatother
\begin{document}
\title{Optimal deployments of UAVs  with directional antennas for a power-efficient coverage}

\author{Jun~Guo, Philipp~Walk, and Hamid~Jafarkhani\\
  {\small\begin{minipage}{\linewidth}\begin{center}
    \begin{tabular}{c}
    Center for Pervasive Communications and Computing \\
    University of California, Irvine, CA 92697-2625 \\
    {\it\{guoj4,pwalk,hamidj\}@uci.edu}\\
    \end{tabular}
    \end{center}
  \end{minipage}\vspace{0.2cm}}
}

\maketitle

\begin{abstract}\blfootnote{This work was supported in part by the NSF Award CCF-1815339. Part of the work was presented
    in Data Compression Conference \cite{GWJ19}.}%
 To provide a reliable wireless uplink for users in a given ground area, one can deploy Unmanned Aerial
 Vehicles (UAVs) as base stations (BSs). In another application, one can use UAVs to collect data from sensors  on the
 ground. For a power-efficient and scalable deployment of such flying BSs, directional antennas can be utilized to
 efficiently cover arbitrary 2-D ground areas. We consider a large-scale wireless path-loss model with a realistic angle-dependent radiation pattern for the directional antennas. Based on such a model, we determine the optimal 3-D deployment of  $N$ UAVs to
 minimize the average transmit-power consumption of the users in a given target area. The users are assumed to have
 identical transmitters with ideal omnidirectional antennas and the UAVs have identical directional antennas with given half-power beamwidth (HPBW) and symmetric radiation pattern along the vertical axis.
 For uniformly distributed ground users,
 we show that the UAVs have to share a common flight height in an optimal power-efficient
 deployment.
 We also derive in closed-form the asymptotic optimal common flight height of $N$ UAVs in terms of the area size, data-rate, bandwidth,  HPBW, and path-loss exponent.

\vspace{-3ex}
\end{abstract}
\section{Introduction}
Due to the decreasing production cost of Unmanned Aerial Vehicles (UAVs), wireless communication coverage for large
areas can be achieved efficiently and flexibly by using a network of UAVs equipped with wireless transceivers. These
UAVs can communicate to each other or to nearby stationary base stations and operate as a relay network for users on the
ground \cite{ZYS11,ZWZ19}.  To improve
wireless links to the users, such flying base stations (BSs) use directional antennas to concentrate the radiation
power to smaller cells on the ground. Hence, directional antennas reduce power consumption and interference with neighboring
cells \cite{BJL,MSF,HA,HSYR,MWMM}. It is common to assume that  the antenna pattern of a directional antenna is
an ideal beam and symmetric in the azimuth plane. In such a model, the radiation intensity is constant  for elevation angles
inside the beam, defined by its beamwidth, and zero or small outside  \cite{YPWS-B19,AE-KLY17,HSYR,MWMM}. Such an
approximation is sufficient for high-altitude UAVs covering small  ground cells, but not for low-altitude UAVs which serve larger cells.
 Moreover, the objective is to find the maximal cell-radius which guarantees a reliable downlink at a given data-rate (coverage). Using the Shannon capacity
formula, for a given bandwidth and noise power, this reduces to a minimal required receive power for each ground user (UE) \cite{AKL14}.
To cover a given target area at the ground, efficiently with $N$ identical UAVs, an optimal common flight height is determined. Because of the circular cell shapes, this approach, in general, does not result in
a full coverage of the target area. By focusing on an uplink coverage and a more realistic model, our approach is slightly different. UEs
can adjust their transmit powers to achieve a reliable uplink connection in a given range and at a given data-rate. Therefore, we consider a full coverage model by using a
transmit-power model which is continuous in the elevation angle and hence continuous in the UE positions. Assuming the UEs  are  distributed by a given density function and for a given uplink data-rate, the objective for an optimal UAV deployment is then to minimize the
average transmit-power over all UEs in the target area
\cite{GJcom18,KKSS18,GWJ19}. The target area can have any polygonal shape which can be fully covered by any number of UAVs.

 To achieve our goal, we use a  more realistic directional antenna pattern, which considers a continuous angle-dependent radiation gain.
Our recent conference paper introduced a similar concept for 2-D UAV deployments to cover 1-D ground areas
\cite{GWJ19}.
 The received UAV power depends on the line-of-sight (LoS) distance between the UAV and the corresponding ground user and the UAV's antenna gain at the corresponding
Angle of Arrival (AoA). As shown in \figref{fig:uavdirected}, the AoA $\tht$ is the arc-cosine of the ratio of the flight height and the LoS distance.
In this manuscript, we extend the model to  3-D deployments and  adjustable beamwidths.
We model the antenna gain by various cosine-powers of the radiation angle (AoA)  \cite{MLW15,Bal05a}.
To minimize the average transmit power of  ground UEs, by deploying $N$ UAVs, a continuous $N-$facility locational optimization problem has to be solved.
 This problem has been  investigated for example in \cite{Erdem16,KJ17,KKSS18,GHHF08} by assuming a given common  UAV flight height and an ideal beam.
We investigate the UAV optimization problem over all possible ground locations  in the target area  and flight heights.
 This results in a  3-D  optimization problem with a parameterized power function.
Such a parameterized cost function can also be used to formulate heterogeneous sensor deployment problems, as for example investigated in \cite{GHHF08}.
In many applications, as in sensor or vehicle deployments, the optimal weights and parameters of the system are usually unknown, but adjustable.
Therefore, one wishes to optimize the deployment over all admissible parameter values \cite{ML}.
In this work, for a large number of UAVs, we derive the closed-form optimal deployments to serve users uniformly distributed in a given  2-D  target area.
For given arbitrary flight heights, the optimal regions (cells) are known to be generalized Voronoi (Mobius) regions, which can be non-convex and disconnected sets \cite{BWY07}.
 A deployment optimization over arbitrary heights constitutes a heterogeneous problem whose solution is not known in a closed-form \cite{KJ17}.
However, our numerical solutions show that asymptotically a common height is optimal.

A dual problem is a downlink scenario, which minimizes the UAVs' average transmit-power to cover UEs at a given average downlink-rate \cite{AE-KLY17}.
Our uplink UAV deployment solution is also optimal for the downlink problem.
The contributions of the paper can be summarized as
\begin{itemize}
    \item We consider a more realistic directional antenna model that considers a continuous angle-dependent radiation gain.
    \item We investigate the optimal 3-D UAV deployment problem over all possible ground locations and flight heights to minimize the total average transmit-power.
    \item We show numerically that the global optimal deployment is asymptotically given by a hexagonal lattice of the UAV ground positions and
a unique common flight height.
\end{itemize}

The rest of the  paper is organized as follows: We introduce a realistic  and mathematically tractable wireless communication model for ground-users-to-UAVs with directional antennas in \secref{sec:model}.
We formulate and solve the optimal  3-D  UAV deployment problem over given arbitrary ground areas in \secref{sec:optimize1D}.
 In \secref{sec:simulations}, we provide iterative Lloyd-like algorithms to derive  UAV deployments for various parameters with uniform and non-uniform user distributions.
Furthermore, we provide simulation results to compare to other deployments derived in \cite{KKSS18,MWMM} and verify the asymptotic optimality of common height deployments.
 Finally, we provide conclusions in \secref{sec:conclusions}.

\paragraph{Notation}  We denote the first $N$ natural numbers, $\Nplus$, by $[N]=\{1,2,\dots,N\}$.  We
write real numbers in $\R$ by small letters and row vectors by bold letters. The Euclidean norm of $\vx$ is given by
$\Norm{\vx}=\sqrt{\sum_n x_n^2}$.
We denote by $\Vor^c$ the complement of the set $\Vor\subset\R^d$. The real  numbers larger than some $a\geq0$ are denoted by
$\R_a$.

\section{System model}\label{sec:model}
We investigate the  3-D  deployment of $N$ UAVs positioned
in $\Ome\times\R_0$, operating as flying BSs to provide a wireless communication link to UEs
in a given  2-D  target region  $\Ome\subset\R^2$ on the ground.
Here, the $n$th UAV's position, $(\vp_n,h_n)$, is given by its ground position $\vp_n=(x_n,y_n)\in\Ome$ and its height $h_n\in\R_0$.
%
The optimal UAV deployment is then defined by the minimum average transmit-power to provide an uplink connection for UEs, distributed by a continuous density function $\lam$ in $\Ome$.
Each UE selects the UAV which requires the smallest  transmit-power\footnote{We assume an orthogonal communication by using frequency or time separation (slotted protocols) with no inter-user interference .}. This results in a so called generalized Voronoi (user) region for each UAV and partitions  $\Ome$ into $N$ user regions.
Hence, the optimal average-power deployment problem of $N$ UAVs is similar to an $N-$point quantization problem, as defined in
\cite{GWJ19,Erdem16,GJ,GJcom18,GJ18,KJ17,ML,MLCS,KKSS18,SJH}.
 For homogeneous deployments, where the BSs are mounted on the ground or at a fixed common height, the Voronoi regions for a large number of BSs converge to the well-known hexagonal regions \cite {OBSC00}.
For heterogeneous BSs or different heights the optimal regions are unknown \cite{KJ17}.

In recent decades, UAVs with directional antennas have been widely studied in the literature to increase the efficiency of wireless links \cite{BJL,MSF,HA,KMR,HSYR,MWMM,YPWS-B19}.
 Usually, the antenna gain $G$ is approximated by a constant within a $3$dB beamwidth, half-power-beam-width  (HPBW) and by zero or a small value  outside the beamwidth, resulting in an ideal directional antenna
pattern
\begin{align}
  G (\tht,\phi)= \begin{cases} \GHPBW , &|\tht| \leq \thtHPBW/2\\ 0, &\text{else}\end{cases}\quad,\quad \tht\in[0,\pi], \phi\in[0,2\pi],\label{eq:GHPBW}
\end{align}
which is symmetric in the azimuth plane.
Such a definition  ignores the strong angle-dependent gain of directional antennas \cite{MLW15}, notably for low-altitude UAVs which serve large user regions.
Since, due to the flight zone restrictions of aircraft's, the maximal heights for UAVs are typically less than $1000$m, such an angle-dependent gain becomes crucial if a few UAVs need to cover large target areas. 
As shown in \figref{fig:uavdirected}, to obtain a more realistic model, we consider an antenna gain that depends continuously on the actual  radiation angle (AoA) $\theta_n(\vome)\in[0,\frac{\pi}{2}]$ from the $n$th UAV at $(\vp_n,h_n)\in\Ome\times\R_0$ to a UE at $\vome\in\Ome$.
To capture the power falloff versus the  Line-of-Sight (LoS)  distance $d_n$ along with the random attenuation due to shadowing, we adopt the following model \cite[(2.51)]{Gol05} %
\setlength{\floatsep} {0pt plus 2pt minus 6pt}
\setlength{\textfloatsep} {12pt plus 2pt minus 6pt}
\begin{figure}[!htb]
\setlength\abovecaptionskip{0pt}
\centering
\includegraphics[width=6in]{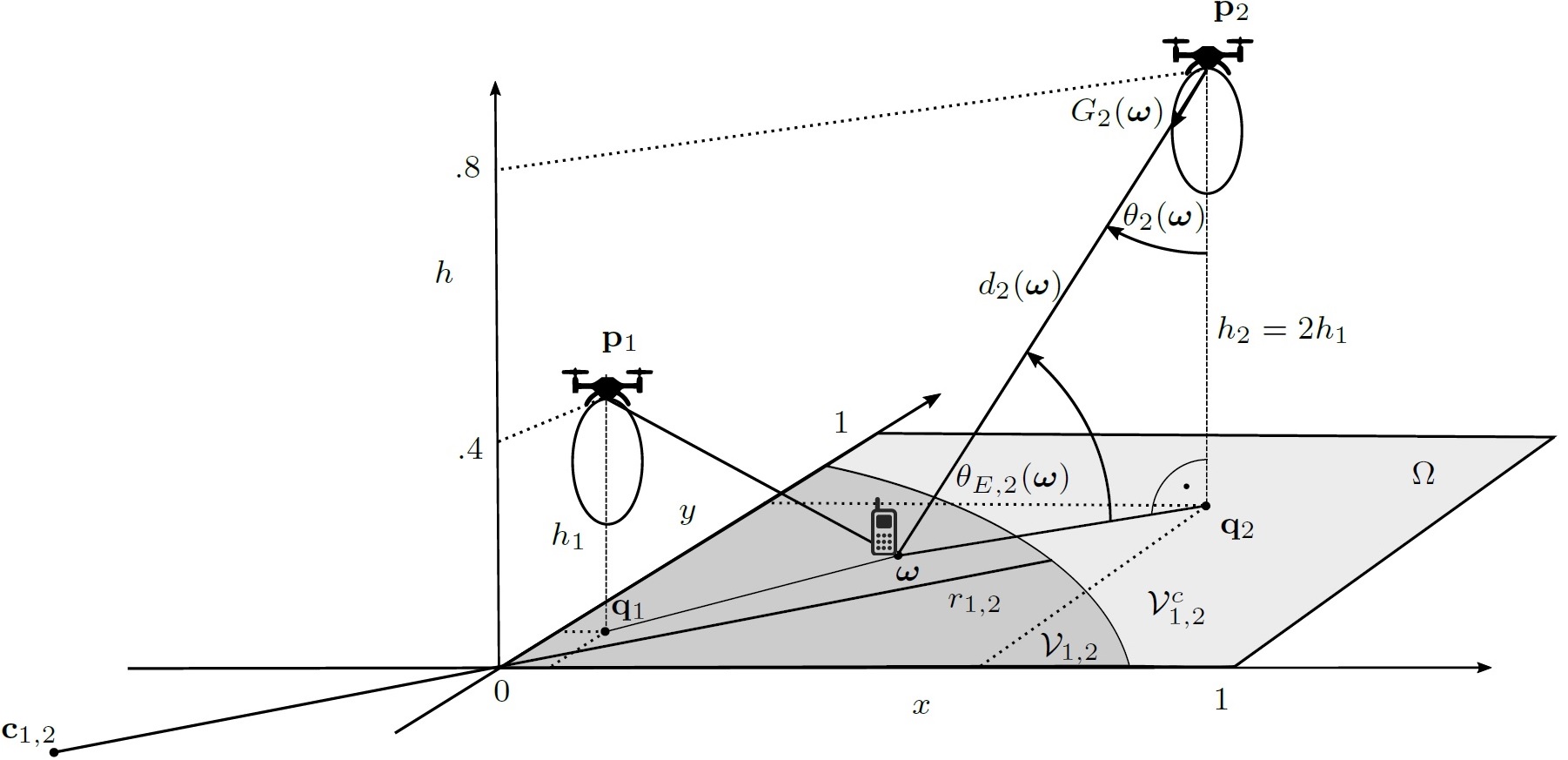}
\captionsetup{justification=justified}
\caption{\small{UAV deployment with directional antenna gains and associated UE cells with path-loss $\alp=2$, antenna
   parameter $\kappa=1$, and $N=2$ UAVs for a uniform UE
 distribution in $\Ome=[0,1]^2$.}}
    \label{fig:uavdirected}
\end{figure}
\begin{equation}
  PL_{dB}=10\log_{10}{K}-10\alpha\log_{10}(d_n/d_0)-\psi_{dB},\label{eq:pathlossmodel}
\end{equation}
where $K$ is a unit-less constant depending on the antenna characteristics and frequency, $d_0$ is a reference distance
to the actual distance $d_n>d_0$  at which an exponential path-loss needs to be considered, $\alpha\geq 1$ is the path-loss exponent, and $\psi_{dB}$ is a
Gaussian random variable following $\mathcal{N}\left(0,\sigma^2_{\psi_{dB}}\right)$ representing the random channel
attenuation (shadowing and non-LoS paths).
This Cellular-to-UAV or terrestrial  log-distance path-loss model is widely used and recommended by both 3GPP and ITU  \cite{ITU09,AG18}.
  Practical values of $\alp$ are between $1$ and $6$. The LoS distance of  UE at $\vome$ to the $n$th UAV at $(\vp_n,h_n)$ is
\begin{align}
 d_n(\vome)=\sqrt{\|\vp_n-\vome\|^2+h_n^2}=\sqrt{(x_n-x)^2+(y_n-y)^2+h_n^2}\label{eq:eucd}.
\end{align}
Common practical measurements of $\alp$ have been provided in \cite{AG18}.
%
With this model, the received power at the $n$th UAV from a UE at $\vome$ is given by  \cite{Gol05}
\begin{align}
    \Prxn(\vome)=\inPtxn(\vome)\beta_n(\vome)=\inPtxn(\vome) \Gtx \Grxn(\vome) Kd^{\alpha}_0 d_n^{-\alpha}(\vome)10^{-\frac{\psi_{dB}}{10}},\label{eq:rxPower}
\end{align}
where $\sqrt{\beta_n(\vome)}$ is the effective channel attenuation between the UE and the UAV.
To derive a realistic channel model, not only do we consider the LoS distance in $\bet_n(\vome)$, but also we take into account the corresponding elevation angle between the UE and the UAV.
For the UE, the dimensionless transmit antenna gain $\Gtx>0$ is assumed to model a perfect omnidirectional (isotropic) antenna, which is identical for all UEs.
The UAVs are equipped with identical directional receive antennas with gains
\begin{equation}
    \Grxn(\vome) = D_0(\kap)\cos^{\kappa}\left(\theta_n(\vome)\right) =D_0(\kap)\frac{h_n^\kap}{ d_n^\kap(\vome)}.
    \label{eq:Gdirected}
\end{equation}
These gains depend on the radiation  angle $\tht=\tht_n(\vome)$ and are symmetric along the vertical direction, i.e., independent of the azimuth angle $\phi$, as for example in horn or uniform linear array (ULA)\ antennas  \cite[Sec.2.6.1]{Bal05a}.
Compare to our conference paper \cite{GWJ19}, we have added an additional antenna parameter $\kappa\geq 1$ to the directional antenna gain which defines the \emph{maximal directivity} of the antenna
\begin{align}
  D_0(\kap)=\frac{4\pi}{\Ome_A(\kap)}\geq 1,\label{eq:maxdirect}
\end{align}
where $\Ome_A(\kap)$ denotes the \emph{beam solid angle} \cite[(2-23)]{Bal05a}.
For simplicity, in the antenna pattern $U_\kappa(\theta)=\cos^{\kappa}(\theta)$, we ignore the $l$ possible minor (side) lobes, which are usually modeled by $\cos(l\tht)$ for a more realistic antenna pattern \cite{MLW15}.
We can ignore the side lobes and especially the back lobes ($|\tht|>\pi/2$) since there is no significant reflection above and side-wards the UAVs when they fly at a reasonable flight height, as shown in   \figref{fig:uavdirected} and \figref{Pattern}.
 In fact, since we are only interested in a power averaged over all user positions in a cell, we essentially average the antenna pattern over all radiation (elevation) angles which is exactly what \eqref{eq:Gdirected} describes.
To account for the power concentration compared to an ideal  isotropic antenna with gain $G_0=1$ in each direction, we
normalize the  symmetric directional antenna gain \eqref{eq:Gdirected}  by the beam solid angle
\begin{align}
  \Ome_A(\kap)=  \int_{0}^{2\pi}\int_{0}^{\pwtwo{\pi}} U_\kap(\tht) \sin(\tht)d\tht d\phi
   = 2\pi\int_0^{\pi/2} \cos^\kap (\tht) \sin(\tht)d\tht= \frac{2\pi}{\kap+1}\quad,\quad \kap\geq1,
   \label{eq:beamsolidangle}
\end{align}
%
%
where the closed-form expression for the last integral is provided in \cite[(2.537.1)]{GR15}. Note that we assumed no back-lobe, i.e., $U_k(\tht)=0$ for $\pi\geq |\tht|\geq \pi/2$.
 For $\kappa=0$, we have an isotropic radiation pattern which results in a beam solid angle (no back reflector) $\Ome_A(0)=4\pi$ and
hence to the directivity $D_0(0)=1$.
The directivity of a directional antenna describes the overall power gain, compared to an isotropic antenna, in the direction of maximal gain ($\tht=0$).
The larger $\kappa$, the larger the directivity of the directional antenna, and the smaller the beam.
Then, large $\kappa$'s model antennas with small beamwidths and allow to focus (collect) the radiation power in a smaller area on the ground (cell), as shown in
\cite[Fig.~4]{MLW15} and \figref{Pattern}.
 A more insightful antenna parameter is given by the beamwidth $\ththpbw$.
The HPBW gain $\Ghpbw$ is by definition \cite{Bal05a} the angle $\ththpbw/2$ at which the gain is half of the maximal gain, i.e. the normalized pattern $U_{\kappa}(\ththpbw/2)\!=\!1/2$. 
Hence, the HPBW relates to $\kappa$ by
\begin{align}
    \ththpbw(\kap)=2\arccos(2^{-1/\kappa}).\label{eq:HPBW}
\end{align}
However, the HPBW only describes the solid angle, in which the gain is at least half-the maximal gain $U_{\kappa}(0)$.
The assumption that most of the radiated power will be radiated in this solid angle leads to the approximation in \eqref{eq:GHPBW}.
Not only does such an approximation  neglect the radiation outside the beamwidth, but also it ignores the fact that the continuous radiation pattern  monotonically decreases in $|\tht|$ over the range $[-\pi/2,\pi/2]$.
By using the gain in \eqref{eq:Gdirected}, we have a mathematically tractable model which respects such a continuous angle dependent radiation gain.
The combined antenna gain is then proportional to $G_{n}(\vome) =K \Gtx \Grxn(\vome)=K \Gtx D_0(\kap)\frac{h_n^\kap}{ d_n^\kap(\vome)}$. 

\begin{figure}
 \centering
  \vskip-3ex
\subfloat[]{%
    \includegraphics[scale=0.3]{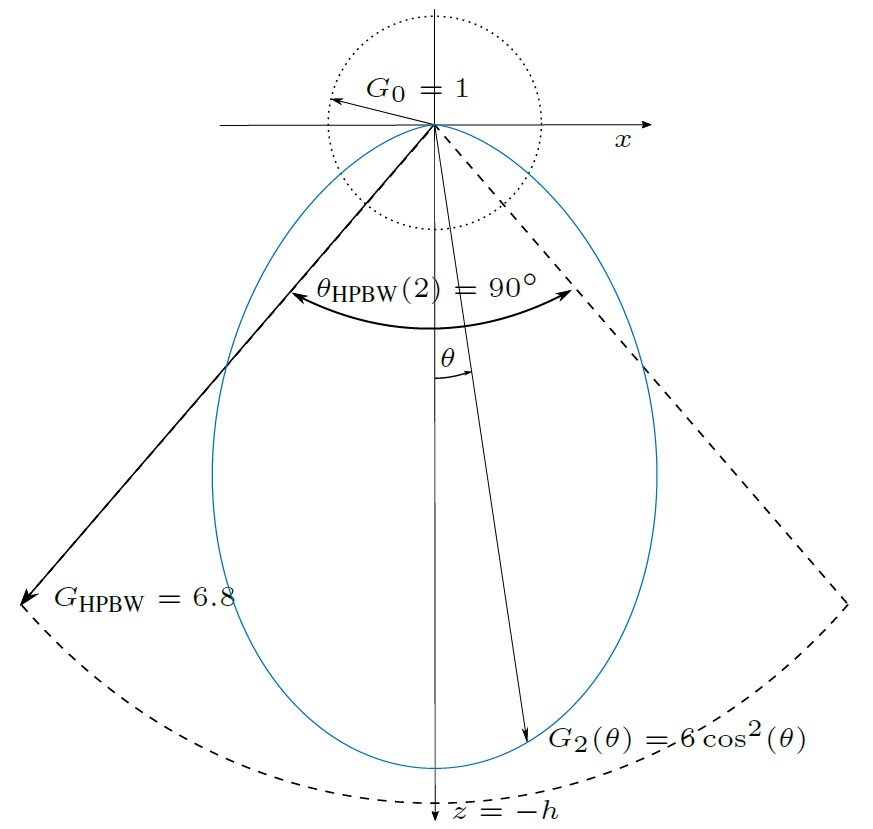}\label{Pattern}
    \label{Pattern}
}
\hspace{-0.16cm}
 \subfloat[]{ 
   \includegraphics[scale=0.25]{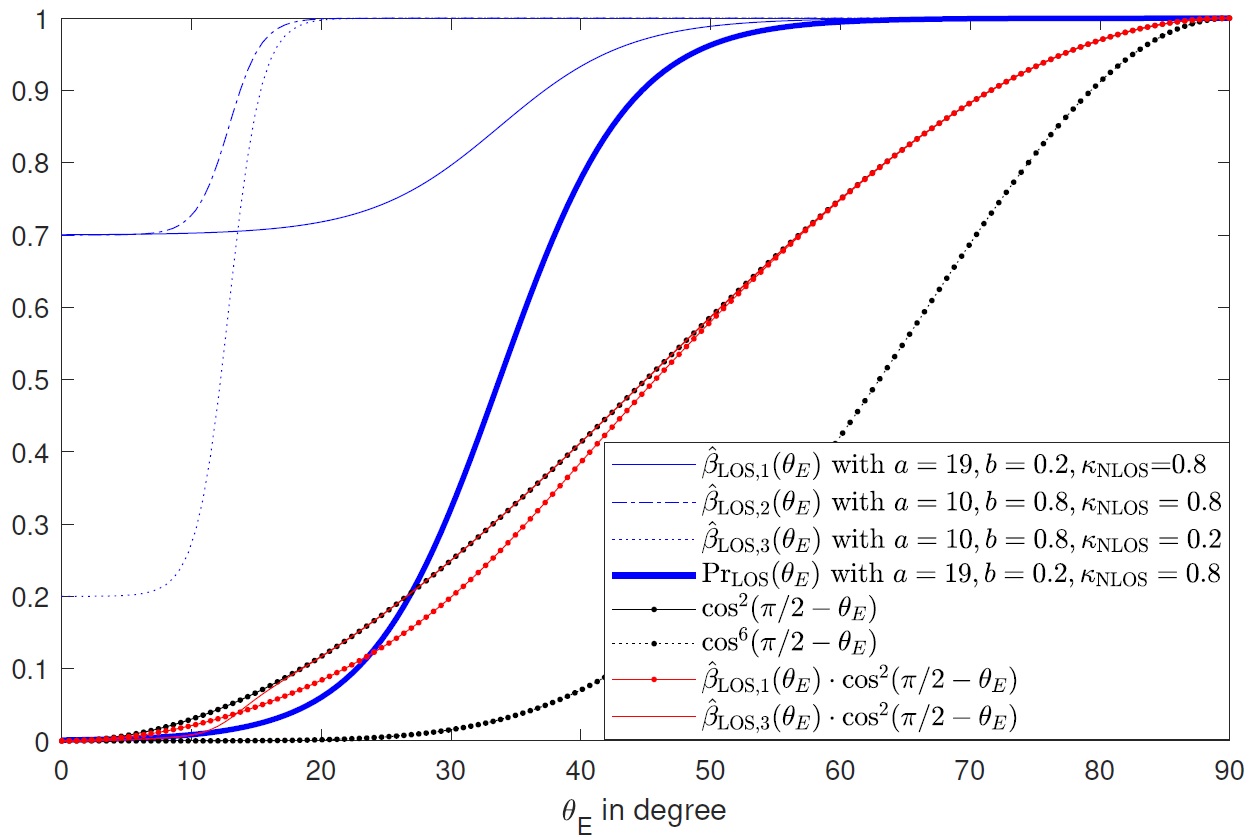}\label{fig:losvsnlos}}
  \caption{{\small \protect\subref{Pattern} Isotropic (dotted), directional with $\kappa=2$ (solid), and
  constant-beamwidth (dashed) antenna gain normalized by directivity \eqref{eq:maxdirect} \pwtwo{ in the elevation plane}. \protect\subref{fig:losvsnlos} shows the attenuation over the elevation angle $\thtelevation$ for some regularized LoS path-loss parameters versus the directional antenna pattern with $\kappa=2$.}} \label{fig:patterns}
\end{figure}

Accordingly, the transmit power with random attenuation $\psi_{dB}$ can be rewritten  as
\begin{equation}
  \inPtxn(\vome)=\frac{\Prxn}{ h_n^\kap K\Gtx D_0(\kap) d^{\alpha}_0}d_n^{\alpha+\kap}(\vome)10^{\frac{\psi_{dB}}{10}}.
\end{equation}
The expectation over the random path-loss attenuation $\psi_{dB}$ yields the transmit power
\begin{align}
  \Ptxn(\vome)=\Expect{\inPtxn(\vome)}&\! =\!
  \frac{\Prxn }{ h_n^\kap K\Gtx D_0(\kap) d^{\alpha}_0}
  \frac{d_n^{\alp+\kap}(\vome)}{\sqrt{2\pi}\sigma_{\psi_{dB}}} \int_{\R}
    \! \exp\!\left(\! \ln(10) \frac{\psi_{dB}}{10}  \!-\!\frac{\psi^2_{dB}}{2\sigma^2_{\psi_{dB}}}\!\right)\!d\psi_{dB}.
     \label{eq:expectedTX} %
\end{align}
We consider the communication between UE and UAV as reliable if the corresponding bit-rate is at least $\Rb$. Given a channel bandwidth $B$ and
 noise power  $N_0$, the Shannon formula suggests that  $\Rb=B\log_2\left(1+ \frac{\Prxn}{N_0}\right)$. Therefore,  the minimum required received power is  $P_0=(2^{R_b/B}-1)N_0$.
The minimum transmit power of UE to achieve a minimum received power of $P_0$ at the $n$th UAV is then given by%
\begin{align}
  \Ptxn(\vome)= \Ptx(\vp_n,h_n,\vome)=\frac{  1}{ \bet_0 }\cdot \frac{1}{D_0(\kap)} \cdot \frac{d_n^{\alp+\kap}(\vome)}{h_n^\kap},
  \label{eq:Eptx}
\end{align}
where the independent and fixed parameters are combined to
\begin{align}
  \bet_0(\alp)=\frac{ K\Gtx  d_0^{\alp}  }{P_0} \exp\left(\frac{\sig_{\psi_{dB}}^2 (\ln 10)^2}{200}\right)
  =\frac{ K\Gtx d_0^{\alp}  \sig_{\psi}^2}{(2^{\frac{\Rb}{B}}-1) N_0 },
  \label{eq:bet}
\end{align}
where $\sig_{\psi}^2$ is the (linear) average-power of the random channel attenuation.
%
%
The first factor in \eqref{eq:Eptx} describes the channel shadowing, noise-power, bandwidth/data-rate, and the antenna characteristics controlled by the path-loss exponent $\alp$. The second factor describes the power gain of the directional antenna with exponent $\kap$ compared to an isotropic antenna with exponent $\kap=0$. A larger $\kap$ results in a larger directivity and a smaller required transmit-power.
The last factor is the angle and distant dependent channel attenuation in $\bet_n(\vome)$ and is the novel part of our model. In this model, the directional antenna gain decreases fast with the radiation angle and punishes large radiation angles, i.e., UEs with a small elevation angle. The main goal of this work is to understand the optimal UAV deployment and optimal directional antenna beam for a given user area and number of UAVs.
For the validity of the path-loss model in \eqref{eq:pathlossmodel}, we need to ensure that $d_n>d_0$ for any UE position $\vome$, which requires a minimum flight height $\hmin>d_0$ for each UAV.
Such a minimum flight height can also be justified from a security point of view, to prevent collisions of the UAV with objects or people on the ground.
Furthermore, a very low UAV height can  result in high transmit-powers for far distant UEs, which might not be  admissible.
As can be seen from \eqref{eq:Eptx}, the transmit-power $\Ptx$ is a function of the parameter $h_n$ (UAV flight height) in addition to the ground distance between $\vp_n$ (UAV ground position) and $\vome$ (UE position).
Hence, a minimization of the average transmit-power for a full coverage of users in $\Ome$ results in a 3-D UAV deployment problem that is solved in the next section.
For simplicity, from now on, we set $\bet_0(\alp)D_0(\kap)=1$ since it does not affect the optimal deployment for fixed $\alp,\kap$.

So far, we have considered a large-scale fading channel model for Cellular-to-UAV links, in which we have included the angle-dependent directional antenna gain.
However, for a ground-user-to-UAV link in an urban area, we also need to consider  small-scale fading with non-LoS paths. Such a non-LoS path is due to the  blockage of objects on the ground, for example by  buildings, trees, or even moving vehicles \cite{MSBD16a,AE-KLY17,AG18,ZXZ19}.
A non-LoS propagation results in a higher path-loss and hence in an additional attenuation of some  $\betnlos\leq 1$.
The probability for a LoS propagation can be approximated in the elevation angle $\thtelevation=\pi/2-\tht$ (measured in radians)  by
\begin{align}
    \Plos(\thtelevation) = \frac{1}{1+a e^{-b(\frac{180}{\pi}\thtelevation -a)}},
\end{align}
for some parameters $b>0$ and $0<a<90$ \cite{AKL14,MWMM,ZWZ19}.
The probability of a LoS path is monotone increasing in the elevation angle and has an $S-$shaped curve, as shown by the thick blue curve in \figref{fig:losvsnlos}.
In more dense urban areas, the non-LoS paths are more likely, even at larger elevation angles, and the S-curve shifts to the right \cite[Fig. 2]{AKL14}.
A probabilistic mixing of LoS and NLoS attenuation results in a regularized LoS path-loss given by the attenuation factor \cite{ZXZ19,ZWZ19}
\begin{align}
   \hPlos(\thtelevation)= \Plos(\thtelevation)+(1-\Plos(\thtelevation))\betnlos = \frac{1+\betnlos a e^{-b(\frac{180}{\pi}\thtelevation-a)}}{1+a e^{-b(\frac{180}{\pi}\thtelevation -a)}}.
   \label{eq:regularizedLoS}
\end{align}
Multiplying  \eqref{eq:rxPower}  by \eqref{eq:regularizedLoS}  yields  $\hPlos(\pi/2-\tht) \cdot\cos^{\kap}(\tht)$ for the linear attenuation factors with omnidirectional transmit antennas.
However, since the directional antenna pattern $\cos^\kap(\theta)$, for some $\kap\geq 1$, decays fast to zero if the AoA $\tht$ approaches $\pi/2$, the antenna intensity dominates the attenuation gain for large $\tht$. As shown in \figref{fig:losvsnlos}, the antenna pattern fully absorbs the regularized LoS attenuation of the S-curve \eqref{eq:regularizedLoS} for large elevation angles. The result is only affected for very small elevation angles (large AoAs).
By choosing a proper minimum height $\hmin$ for a given area $A=|\Omega|$, large AoAs can be avoided, such that we can neglect the NLoS path effects.
Furthermore, we could also add an additional antenna exponent $\kapnlos>0$ to $\kap$ to increase the descent of the Cosine, which can approximate the regularized NLoS path in denser urban areas, as shown in  \figref{fig:losvsnlos} for $\kapnlos+\kap=4+2=6$.

\section{Optimal UAV Deployments} 
\label{sec:optimize1D}

The transmit power \eqref{eq:Eptx} defines, with $h_n$ and fixed $\alp,\kap\geq 1$, a
parameter-dependent power function for $\vp_n$. For a given UE density $\df$ in $\Omega$, UAV deployment
$(\vP,\vh)$ with ground positions $\vP=(\vp_1,\dots,\vp_N)$, heights $\vh=(h_1,\dots,h_N)$, and user regions (cells)
$\Rset=\{\Rset_1,\dots,\Rset_N\}$ with $\bigcup \Rset_n=\Ome$, the average transmit power $\AvDis$ of each UE in $\Ome$
for $\gam=\frac{\alp+\kap}{2}\geq 1$ is given by
\begin{equation}
 \AvDis(\vP,\vh,\Rset)
 = \sum_{n=1}^N \int_{\Rset_n} P(\vome,\vp_n,h_n)\lam(\vome)d\vome \ \text{  with  }\
 P(\vome,\vp_n,h_n)=\frac{(\Norm{\vome-\vp_n}^2+h_n^2)^{\gam}}{h_n^\kap}.
  \label{eq:Pbar}
\end{equation}
Here, we assume that the UE at $\vome$ transmits with the smallest power $P$ to achieve a reliable link to the nearest UAV at $(\vp_n,h_n)$.
The $N$ regions, which minimize the average transmit power for given ground positions and heights
$(\vP,\vh)$, define a generalized Voronoi tessellation $\Vor=\{\Vor_n(\vP,\vh)\}$ of $\Ome$ by
\begin{align}
  \AvDis(\vP,\vh)
  :=\!\int_{\Omega}\min_{n\in[N]} \left\{ \Dis(\vome,\vp_n,h_n) \right\} \df(\vome)d\vome
  =\sum_{n=1}^{N}\int_{\Vor_n(\vP,\vh)}\!\!\!\!\!\Dis(\vome,\vp_n,h_n) \df(\vome)d\vome
  \label{eq:optPbar},
\end{align}
where the \emph{generalized Voronoi regions} $\Vor_n(\vP,\vh)$ are defined as the set of sample points (user positions)
$\vome$ with smallest power to the $n$th ground position $\vp_n$ with parameter $h_n$ (UAV position).  Minimizing the
\emph{average transmit-power} $\AvDis(\vP,\vh,\Vor)$ over all UAV positions can be seen as an
\emph{$N-$facility locational-parameter optimization problem} \cite{GJ, GJcom18, GJ18,OBSC00}. According to the definition of the
Voronoi regions in \eqref{eq:optPbar}, we have
\begin{align}
  \Vor_n(\vP,\vh)=\set{\vome\in\Ome}{P(\vome,\vp_n,h_n)\leq P(\vome,\vp_m,h_m) \text{ for all } m\not=n}.
\end{align}
The \emph{minimum average transmit-power} over all possible 
deployments is then given by
\begin{align}
 \AvDis^*= \AvDis(\vP^*,\vh^*)
  = \min_{(\vP,\vh)\in\Ome^N\times\R_+^N} \AvDis(\vP,\vh)
  = \min_{(\vP,\vh)\in\Ome^N\times\R_+^N} \min_{\Rset=\{\Rset_n\}\subset\Ome} \AvDis(\vP,\vh,\Rset).
\label{eq:optquanteqoptdeploy}
\end{align}
%
%
%
%
To find the local extrema of \eqref{eq:optPbar} analytically, we need the objective function $\AvDis$ to be continuously
differentiable at any point in $\Omega^N\times \R_0^N$, i.e., the gradient should exist and be a continuous function.
Such a property was shown to be true for piecewise continuous non-decreasing cost functions with Euclidean
metrics over $\Ome^N$ \cite[Thm.2.2]{CMB05} and weighted Euclidean metrics \cite{GJ}. Then, the necessary condition for a
local extremum is the vanishing of the gradient at a critical point\footnote{If $\nabla \Pbar$ is not
continuous in $\Qset^N$, then any jump-point is a potential critical point and has to be checked individually.}.
In the next lemma, we derive the generalized Voronoi regions for sets $\Ome\subset\R^d$ with $d=1,2$ and for any height
$h_n\in\R_0$. The derived generalized Voronoi regions  which are special cases of \emph{M{\"o}bius diagrams (tessellations)},
introduced in \cite{BWY07}.
\begin{lemma}\label{lem:moebiusdia}
  Let $\vP=(\gp_1,\gp_2,\dots, \gp_N)\in \Omega^N\subset (\R^d)^N$ for $d\in\{1,2\}$ be the ground positions
  and $\fH\in\R_+^N$  the associated heights. For fixed parameters $\kap\geq 1$ and $\gam\geq
  \frac{1+\kap}{2}$ with uniform density $\df$ in $\Ome$,  the minimal average power over all
  possible $N$ regions is given by
  \begin{align}
    \AvDis\left(\vP,\fH\right)
    = \sum_{n=1}^{N} \int_{\Vor_n} \! \frac{ (\Norm{\gp_n- \vome}^2 +h_n^2)^{\gam}}{h_n^\kap} \df(\vome)d\vome,
       \label{eq:minimizationmoebius}
  \end{align}
  where the generalized Voronoi regions $\Vor_n= \Vor_{n}(\gP,\fH)= \bigcap_{m\not=n} \Vor_{nm}$ and
  the dominance regions of $n$ over $m$ is defined by
  \begin{align}
    \Vor_{nm}=\begin{cases}
         \set{\vome\in\Omega}{\|\gp_n-\vome\|\le \|\gp_m-\vome\|}&, h_m=h_n, \\
         \set{\vome\in\Omega}{\|\vome-\vc_{nm}\|\le r_{nm}} &, h_n<h_m, \\
         \set{\vome\in\Omega}{\|\vome-\vc_{nm}\|\ge r_{nm}} &, h_n>h_m,
        \end{cases}\label{eq:moebius}
  \end{align}
  where center $\vc_{nm}$ and radius $r_{nm}$ of the ball are given by
  \begin{align}
    \vc_{nm}\!=\!\frac{\gp_n - h_{nm}\gp_m}{1-h_{nm}}
    \quad\text{and}\quad
    r_{nm}\!=\!\left(\frac{h_{nm}}{\left(1-h_{nm}\right)^2}\Norm{\gp_n-\gp_m}^2  + h_n^2
    \frac{h_{nm}^{1-\frac{2\gam}{\kap}}
  -1}{1-h_{nm}}\right)^{\frac{1}{2}}.
  \label{eq:rnmcnm}
  \end{align}
  Here, we denoted the height ratio of the $n$th and $m$th UAV by $h_{nm}= \left(h_n/h_m\right)^{\frac{\kap}{\gam}}$.
\end{lemma}

\begin{proof}
  See \appref{app:moebiusregions}.
\end{proof}

\begin{remark}
  It is also possible that two UAV ground positions are the same, but have different flight heights. If the height
  ratio is very small or very large, one of the UAVs can become redundant, i.e.,  its  region is empty, as shown in the following example.
  In fact, if we optimize over all UAVs, such a case will be excluded.  We showed this fact for the optimal 2-D deployment
  in \cite[Lem.3]{GWJ18a}.
\end{remark}

\begin{example}
\figref{fig:uavdirected} plots the UE regions for a
uniform distribution in $\Omega=[0,1]^2$ and UAVs placed on $\vp_1=(0.1 , 0.2),\ h_1=0.5$ and $\vp_2=( 0.6 , 0.6),\ h_2=1$, with parameters $N=2$, $\kap=1$, and $\alp=2$.
If the second UAV reaches an altitude of $h_2\geq 2.3$, its Voronoi region $\Vor_{2}=\Vor_{2,1}$ will be empty and hence it becomes ``inactive''.
\end{example}
\vspace{-20pt}
\subsection{Necessary optimal conditions}
\vspace{-5pt}
To find the optimal deployment of $N$ UAVs, we have to minimize the average transmit power
\eqref{eq:Pbar} over all possible UAV positions  with minimal flight height $h_{min}$, i.e., we have to solve the following non-convex
\emph{$N-$facility locational-parameter optimization problem} \eqref{eq:optquanteqoptdeploy}
\vspace{-3pt}
\begin{align}
  \AvDis(\vP^*,\vh^*)= \min_{\vP\in\Omega^N,\bH\in \RhminN} \sum_{n=1}^{N} \int_{\Vor_n(\vP,\bH)}
  h_n^{-\kap}(\Norm{\vp_n-\vome}^2 +h_n^2)^{\gam}\df(\vome)d\vome\label{eq:phoptlocal},
\end{align}
where $\Vor_n(\vP,\vh)$ are the M繹bius regions given in \eqref{eq:moebius} for each fixed $(\vP,\vh)$. The integral
kernel $f(\Norm{\vp_n-\vome}^2,h_n)=(\Norm{\vp_n-\vome}^2+h_n^2)^\gam/h_n^\kap$ is a non-decreasing positive function in
the Euclidean distance of $\vp_n$ and $\vome$ for fixed $h_n$, since $\gam\geq (1+\kap)/2\geq 1$.
A point $(\vP^*,\vh^*)$ with M繹bius diagram $\Vor^*=\Vor(\vP^*,\vh^*)=\{\Vor_1^*,\dots,\Vor^*_N\}$ is a critical point
of \eqref{eq:phoptlocal} if all horizontal partial derivatives of $\AvDis$  and
the vertical partial derivatives of $\AvDis$ are vanishing
, i.e., if for each $n\in[N]$ we have
\cite{GJ}
\vspace{-3pt}
\begin{align}
  &\nabla_{\vp_n} \AvDis\Big|_{\vp_n=\vp^*_n} = \frac{2\gam}{h_n^{*\kap}}\int_{\Vor^*_n} (\vp^*_{n}-\vome) (\Norm{\vp^*_n-\vome}^2
  +h_n^{*2})^{\gam-1} \df(\vome)d \vome=\zero,\label{eq:pnopt}
\end{align}
\begin{equation}
    \begin{aligned}
\nabla_{\!h_n} \AvDis \Big|_{h_n=h^*_n}
\!\!\!&=\!\frac{\kap}{h_n^{*\kap+1}}\!\!\int_{\Vor^*_n}\! \!\!\left(\!\frac{2\gam h_n^{*2}}{\kappa} ( \Norm{\vp_n^*\!-\!\vome}^2\!+\!h^{*2}_n)^{\gam-1} \!-\!(\Norm{\vp_n^*\!-\!\vome}^2\!+\!h^{*2}_n)^{ \gam}\!\right) \!\df(\vome)d\vome=0.
  \label{eq:hopt}
\end{aligned}
\end{equation}
 If $h_n^*<\hmin$, then we set $h_n^*=\hmin$. In this case, the global optimal deployment is not admissible and $(\vP^*,\vh^*)$ becomes a local optimum. However, it is possible to achieve a global optimum, by adjusting the parameters $\alpha,\kappa,N,$ and $\Ome$ accordingly.
%
%
%
For $N=1$, the integral regions do not depend on $\vP$ or $\vh$ and since the integral kernel $f$ is continuously
differentiable and non-decreasing in $\Norm{\vp_n-\vome}^2$,  the partial derivatives only apply to the integral
kernel \cite{DFG99}. For $N>1$, the conservation-of-mass law \cite{CMB05} can be used to show that the
derivatives of the integral domains cancel each other, see \cite{GJ} for a detailed proof.
\begin{remark}
  The shape of the regions depend on the UAV  heights. If the height is different for each UAV (heterogeneous), some region boundaries will be spherical and not polyhedral. We show later that homogeneous (common) heights with polyhedral
  regions are the optimal regions.
\end{remark}

\subsection{Optimal common height in a $3$-D UAV deployment}\label{sec:twoD}
\newcommand{\Hexa}{\ensuremath{\mathcal{H}}}
\renewcommand{\vq}{\ensuremath{\mathbf{q}}}
\renewcommand{\vQ}{\ensuremath{\mathbf{Q}}}
\newcommand{\MDelepsH}{\ensuremath{M_{\Del}(\eps,H)}}
\newcommand{\MDeleps}{\ensuremath{M_{\Del}(\eps)}}
\newcommand{\MDel}{\ensuremath{M_{\Del}}}
\newcommand{\MtHexa}{\ensuremath{\tM_{\Hexa}}}
\newcommand{\MHexa}{\ensuremath{M_{\Hexa}}}


We showed in the conference paper version of this work \cite{GWJ19} that UAVs in an optimal 2-D deployment achieve a common flight height in the asymptotic limit ($N\to\infty$). Therefore, let us  assume a common height for all UAVs in the 3-D deployment.
 We will show in \secref{sec:simulations} by simulations that the optimal deployment indeed converges to a common height deployment for large $N$.
  For any common height, \lemref{lem:moebiusdia} shows that the ground regions are polyhedral, since the cost function is homogeneous.
Hence, in the asymptotic limit, the optimal ground positions result in congruent hexagonal regions given a fixed common height  \cite{DFG99}.

In fact, we can show that the $3$-D UAV deployment problem with  a common height restriction and for large $N$ has only one local optimum
deployment, given by centroidal ground positions with hexagonal regions and an optimal common height $h^*$.
\begin{theorem}\label{thm:2D}
  Let $\kap\geq1, \gam\geq (1+\kap)/2$, and  $\Ome\subset \R^2$ be a set with area $\mu(\Ome)=A>0$.  For a
  uniform density $\lam$ over $\Ome$, the optimal deployment of $N$ UAVs, minimizing the average  transmit-power  in
  \eqref{eq:phoptlocal} under the restriction of a common height, with ground locations $\vQ^*=(\vq^*_1,\dots,\vq^*_N)$
  and common height $h^*$ is attained asymptotically ($N\to \infty$, high resolution case) by the hexagonal lattice,
  where each Voronoi region $\Vor^*_n$ is congruent to the hexagon $\Hexa$ and the optimal ground locations are the corresponding centroids. 
  Moreover, the global optimal common height is given by $h^*=h^*(\gam,\kap,H)$ for $H=A/N$.  For $\gam=1,2,3,$ we
  derive the optimal height asymptotically as 
  \vspace{-2pt}
  \begin{align}
    h^*(\gam,\kap,H)\sim c(\gam,\kap)\sqrt{H}\quad \text{ for }\quad 1\leq\kap\leq2\gam-1,\label{eq:optcommonheight}
  \end{align}
  \vspace{-2pt}
  with scaling factors
  \begin{align}
    c(1) &=  \sqrt{\frac{5}{18\sqrt{3}}}, \quad
    c(2,\kap)= \sqrt{\frac{5}{18\sqrt{3}} \frac{\sqrt{(172-43\kap)\kap/125 +4}-(2-\kap)}{4-\kap}},\\
 c(3,\kap) &
 =\sqrt{\frac{5}{18\sqrt{3}} \frac{ (u(\kap)-v(\kap))^{\frac{1}{3}} + (u(\kap)+v(\kap))^{\frac{1}{3}} -(4-k) }
  {6-\kap}},
 \label{eq:optcomonheights}
 \end{align}
 where
 \vspace{-5pt}
 \begin{align}
   u(\kap) & =(143360 - 16728\kap -444 \kap^2 + 37\kap^3)/4375,\\
   v(\kap) & =\frac{12(6-\kap)}{125\cdot 35}\sqrt{\frac{3}{5}}
   \sqrt{ 6607552\!+\! 659680 \kap \!+\! 103387 \kap^2\! -\! 108408\kap^3 \!+\! 9034\kap^4},
 \end{align}
which achieves for $\bet_0=1$  and directivity $D_0(\kap)$ in \eqref{eq:maxdirect}  the minimal average transmit powers
 \begin{align}
    &\AvDis^*(1,\kap,H) \sim \frac{1}{D_0(1)}\sqrt{\frac{2}{9\sqrt{3}}}H^{\frac{1}{2}}, \quad
    \AvDis^*(2,\kap,H) \sim \frac{1}{D_0(\kap)}\!
    \left(\frac{14}{405c(2,\kap)} \!+ \!\frac{5c(2,\kap)}{9\sqrt{3}} \!+ \!c^3(2,\kap)\right)\! H^{\frac{3}{2}},\notag\\
 &\AvDis^*(3,\kap,H) \sim \frac{1}{D_0(\kap)}
  \left( \frac{83}{195\cdot 27 c(3,\kap)} + \frac{14c(3,\kap)}{135} + \frac{5c^3(3,\kap)}{9\sqrt{3}} +
  c^5(3,\kap)\right) H^{\frac{5}{2}}.\label{eq:optimalaveragepower}
  \end{align}
\end{theorem}
\begin{proof}
  See \appref{sec:proof2D}
\end{proof}
\begin{remark}
  From simulations, we can find that, for fixed $\kap$,  the factor $c(\gam,\kap)$ is decreasing in $\gam\geq 1$  and
  hence the optimal height for fixed $H$, see \figref{fig:2Doptimalheightk}. Similar results hold for the 2-D deployment and can be shown analytically  \cite[Thm.1]{GWJ18a}.

  For a realistic path-loss model of the wireless link in \eqref{eq:pathlossmodel}, the ratio $H=A/N$ and the parameters $\alp$ and $\kap$ have to be
  chosen such that the optimal height satisfies $h^*>d_0$ and  $h^* \geq \hmin>d_0$ given a minimal height constraint $\hmin$.
  Note that if $d_0=1$, then $\bet_0$ in \eqref{eq:bet} is independent of $\alp$.
\end{remark}
\figref{fig:D1D2} and  \figref{fig:D1D2D3kappa} show the minimal average transmit-power in dB over
$H=A/N$ for $\kap=1$ and over various $1\leq\kap\leq 2\gam-1$ and $\gam=1,2,3$ corresponding
to $\alp=2\gam-\kap$, respectively.
Note that, in \figref{fig:D1D2D3kappa}, an increase in $\kap$ decreases $\alp$ since $\gam$ is fixed.
Since the directivity $D_0=2(\kap+1)$ increases with $\kappa$, it reduces the average transmit-power in \eqref{eq:optimalaveragepower}. For $\alp=1$ and $\kap=1$ (solid curve), we gain $3$dB if we use $\kap=3$ (dashed curve).
Hence, in \figref{fig:D1D2D3kappa}, $\AvDis^*(2,3,100)$ is the smallest
average transmit power for $H=100$, given for example by $N=100$ UAVs covering  $\Ome=[0,100]^2$.

Let us note that all values of $H=A/N$ can be achieved in the high resolution case by an appropriate choice of the  target size $A$.
%
%
An increase of $\kap$ for fixed $\alp$ can change the transmit-power, but only slightly changes the
optimal common height. Hence, the optimal $\kappa$ can be determined to minimize the
transmit-power by optimizing the effective beamwidth.

\begin{figure}
\setlength{\floatsep} {0pt plus 2pt minus 6pt}
\setlength{\textfloatsep} {20pt plus 2pt minus 6pt} 
\begin{center}
  \subfloat[]{
    \includegraphics[scale=0.22]{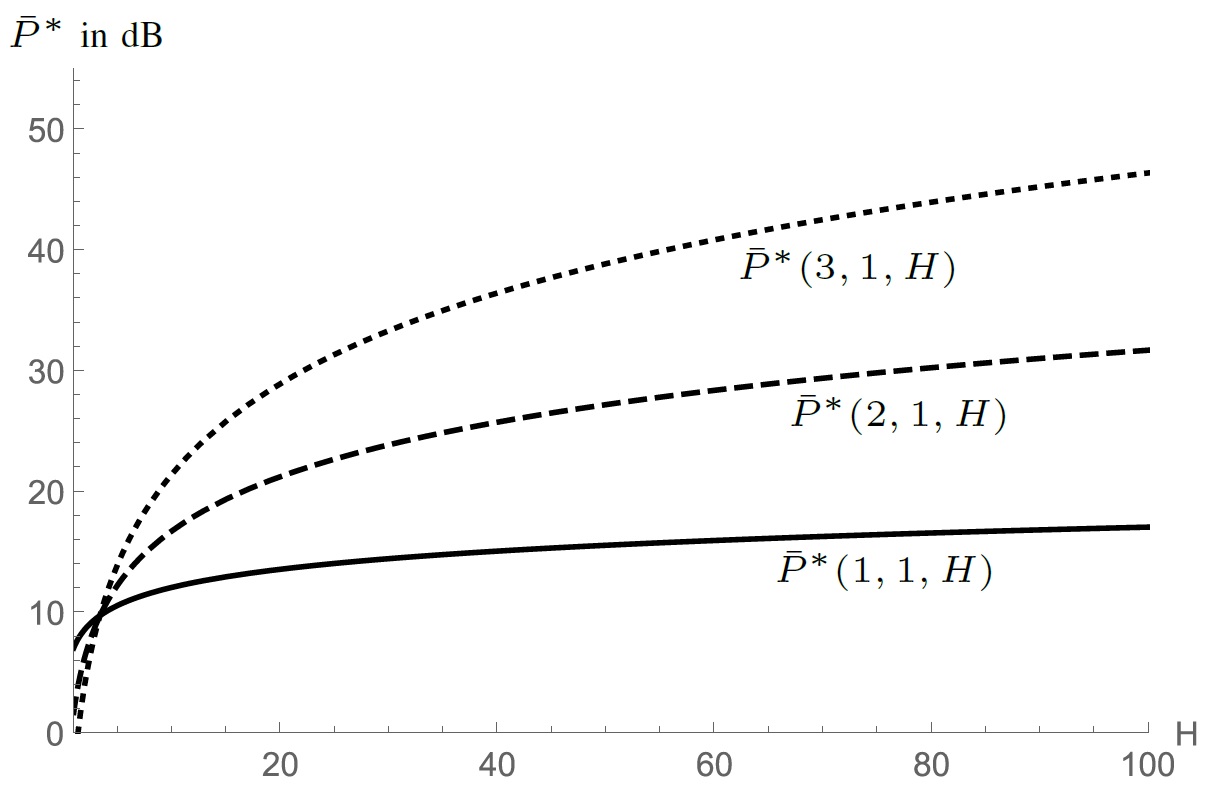}
    \label{fig:D1D2}}
  \hspace{1cm}
  \subfloat[]{
    \includegraphics[scale=0.22]{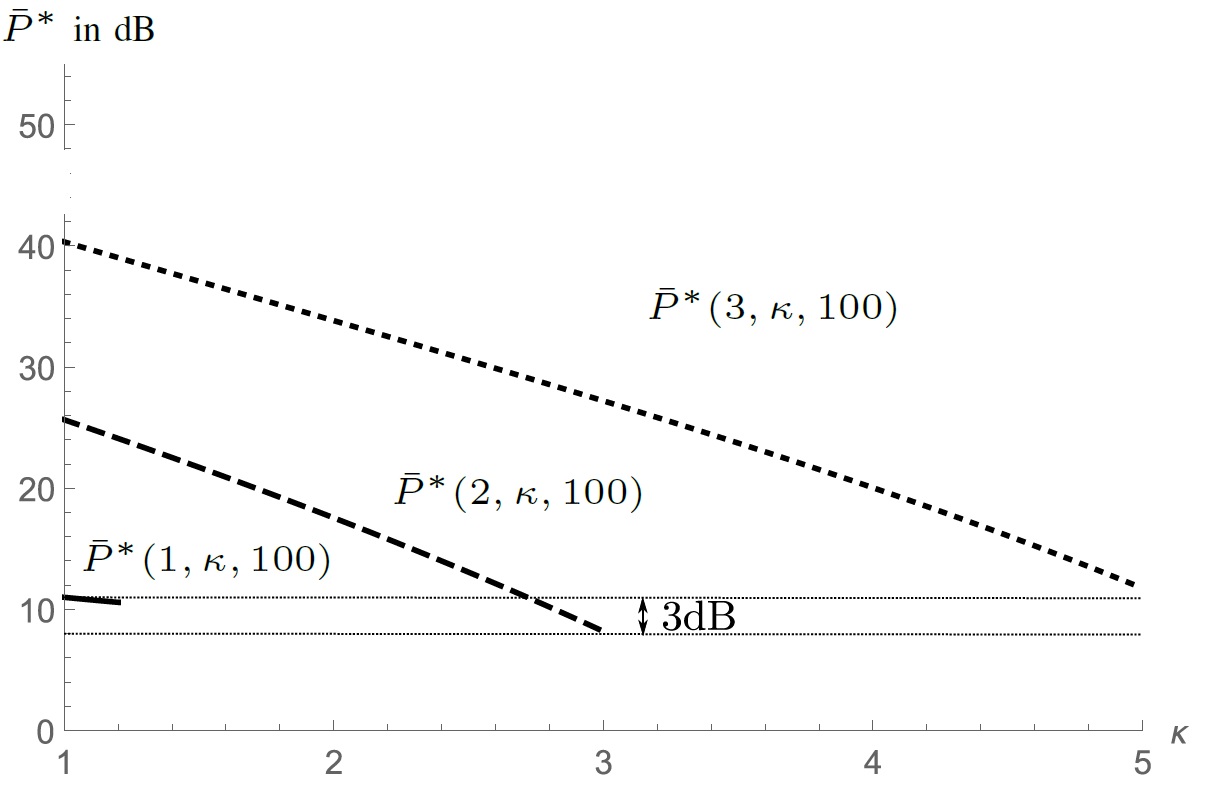}
    \label{fig:D1D2D3kappa}}
\end{center}
\captionsetup{justification=justified}
\vskip-3ex
\caption{\small{Optimal average-powers for $\gam\!=\!1,2,3$ {\protect\subref{fig:D1D2}} over cells of size $H\!=\!A/N$ for $\kap=1$ and {\protect\subref{fig:D1D2D3kappa}} over various $\kap$ for $H=100$.}}
\label{fig:commonheight2d}
\end{figure}

From simulations, for a uniform user density, in the next section, we see that restricting the heights to be the same, i.e., a common height, does not lead to smaller average transmit-power if $N$ is very large. Therefore, we conjecture, that the optimal
common height with centroidal ground positions achieves the global minimal average transmit-power.
\section{Llyod-like Algorithms and Simulation Results}\label{sec:simulations}
In this section, we introduce two Lloyd-like algorithms, Lloyd-A and Lloyd-B, to optimize the deployment for 3-dimensional scenarios. The proposed
algorithms iterate between two steps: (1) The UAV positions are optimized through gradient descent while the
ground cell partitioning is fixed; (ii) The partitioning is optimized while the UAV positions are fixed.  In Lloyd-A, all UAVs
 share the common flight height while Lloyd-B allows UAVs with different flight heights. More details can be found in Algorithm \ref{LloydA}.
\begin{algorithm}[!htb]
\setlength\abovecaptionskip{0pt}
\caption{Lloyd-like Algorithms (Lloyd-A and Lloyd-B)}
\label{LloydA}
\begin{algorithmic}[1]
\INPUT
Target area: $\Omega$;
probability density function: $\lambda(\cdot)$;
the initial UAV ground deployment: $\vP=(\vp_1,\vp_2,\dots,\vp_N)$;
the initial UAV heights: $\vh=(h_1,h_2,\dots,h_N)$ ( $h_1=h_2=\dots=h_N$ for Lloyd-A);
path loss parameter: $\alpha$;
the  antenna pattern exponent: $\kappa$;
the minimum flight height: $h_{min}$;
the initial step size: $\delta$;
the stop threshold: $\epsilon$.
\OUTPUT
the final UAV ground deployments $\vP=(\vp_1,\vp_2,\dots,\vp_N)$;
the final flight height: $\vh=(h,\dots,h)$ for Lloyd-A or $\vh=(h_1,h_2,\dots,h_N)$ for Lloyd-B;
Total  average transmit-power  at the final deployment $\AvDis(\vP,\vh)$.
\State Calculate the generalized Voronoi regions $\Vor_n$, $\forall n\in\{1,\dots,N\}$
\Do
\State Calculate the old total power $\AvDis_{old}=\AvDis(\vP,\vh)$
\State Calculate the gradient $\nabla_{\vp_n}$ and $\nabla_{\vh_n}$ by (\ref{eq:pnopt}) and (\ref{eq:hopt})
\State Initialize step size $t=\delta$
\If {($\nabla_{\vp_n}\ne0$ or $\nabla_{\vh_n}\ne0$, $\forall n\in\{1,\dots,N\}$)}
\Do
\State Calculate the new ground positions $\vp'_n=\vp_n-t*\nabla_{\vp_n}$
\State Calculate the new heights
$\begin{cases}
h'_n=max(h_{min},h_n-t*\sum_{n}\nabla_{\vh_n}), & \text{Lloyd-A}\\
h'_n=max(h_{min},h_n-t*\nabla_{\vh_n}),  &\text{Lloyd-B}\\
\end{cases}$
\State Adjust the step size $t=t/2$
\doWhile{$\AvDis(\vP,\vh)\le \AvDis(\vP',\vh')$}
\EndIf
\State Update UAV deployment $\vP=\vP'$, $\vh=\vh'$
\State Update the generalized Voronoi regions $\Vor_n$, $\forall n\in\{1,\dots,N\}$
\State Calculate the new total power $\AvDis_{new}=\AvDis(\vP,\vh)$
\doWhile{$\frac{\AvDis_{old}-\AvDis_{new}}{\AvDis_{old}}>\epsilon$}
\end{algorithmic}
\end{algorithm}

In what follows, we provide the simulation results over the two-dimensional target region $\Omega=[0,1000]^2$ with
uniform and non-uniform density functions.  The non-uniform density function is a Gaussian mixture of the form
$\sum_{k=1}^{3}\frac{A_k}{\sqrt{2\pi}\sigma^2_k}\exp{\left(-\frac{\|\vome-\vc_k\|^2}{2\sigma_k}\right)}$, where the
weights, $A_k$, $k=1,2,3$ are $0.5$, $0.25$, $0.25$, the means, $\vc_k$, are $(300,300)$, $(600,700)$, $(750,250)$, the standard deviations,
$\sigma_k$, are  $1.5$, $1$, and $2$, respectively.
 All length parameters are measured in meters. The power values depend on the fixed parameters defining $\bet_0$ in \eqref{eq:bet}, given by the bandwidth $B$, data-rate $R_b$, noise-power $N_0$ at the UAV, channel attenuation power $\sig_{\psi}^2$, antenna characteristic $K$, UE antenna gain $\Gtx$ ( $=1$ for perfect isotropic antennas), and the reference distance $d_0^\alpha$ (usually set to $1$). The beam-exponent $\kappa$ and path-loss exponent $\alpha$ are dimensionless and can be adjusted to obtain optimal heights and average-transmit powers 
in a desired range. In our simulations, we use $\bet_0=1$.

\setlength{\floatsep} {0pt plus 2pt minus 6pt}
\setlength{\textfloatsep} {2pt plus 2pt minus 1pt} 
\begin{figure}[!htb]
\setlength\abovecaptionskip{0pt}
\centering
\subfloat[]{\includegraphics[width=2in]{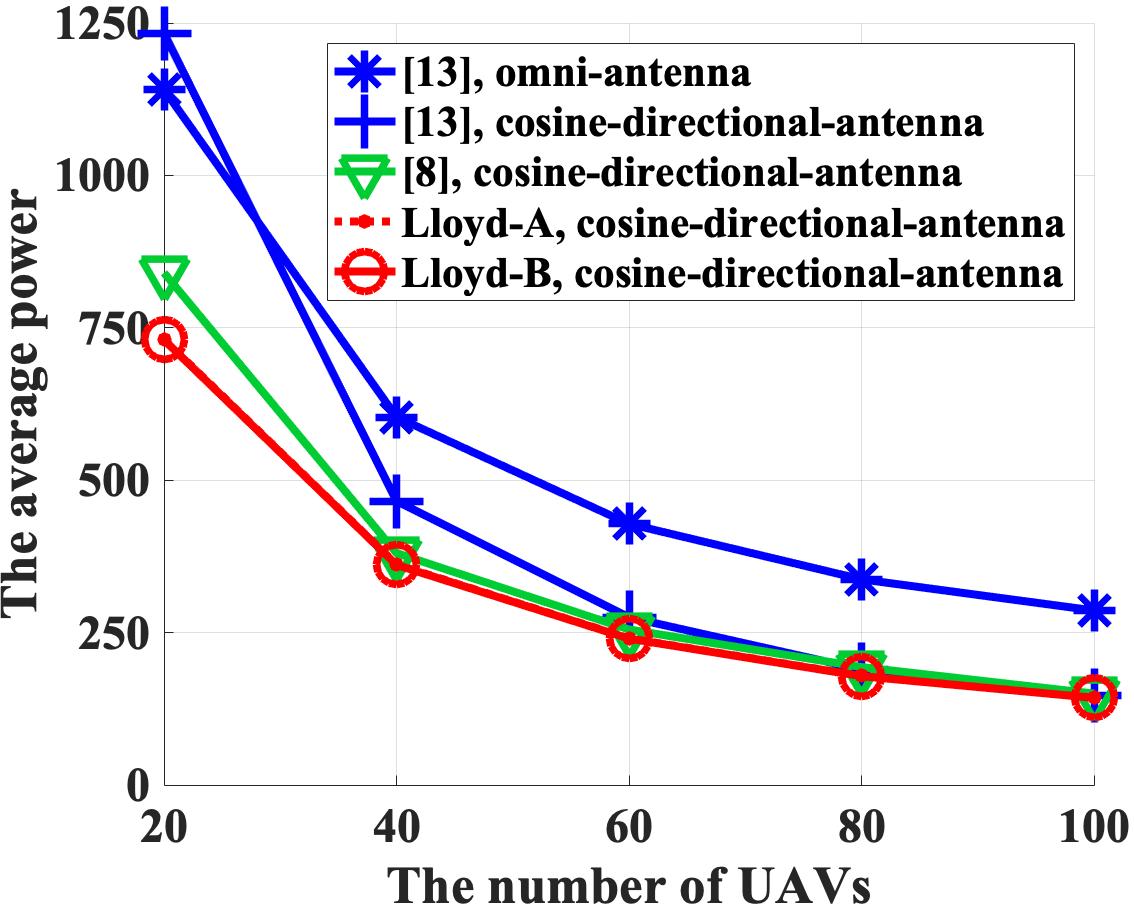}
\label{Alpha2MinH25}}
\hfil
\subfloat[]{\includegraphics[width=2in]{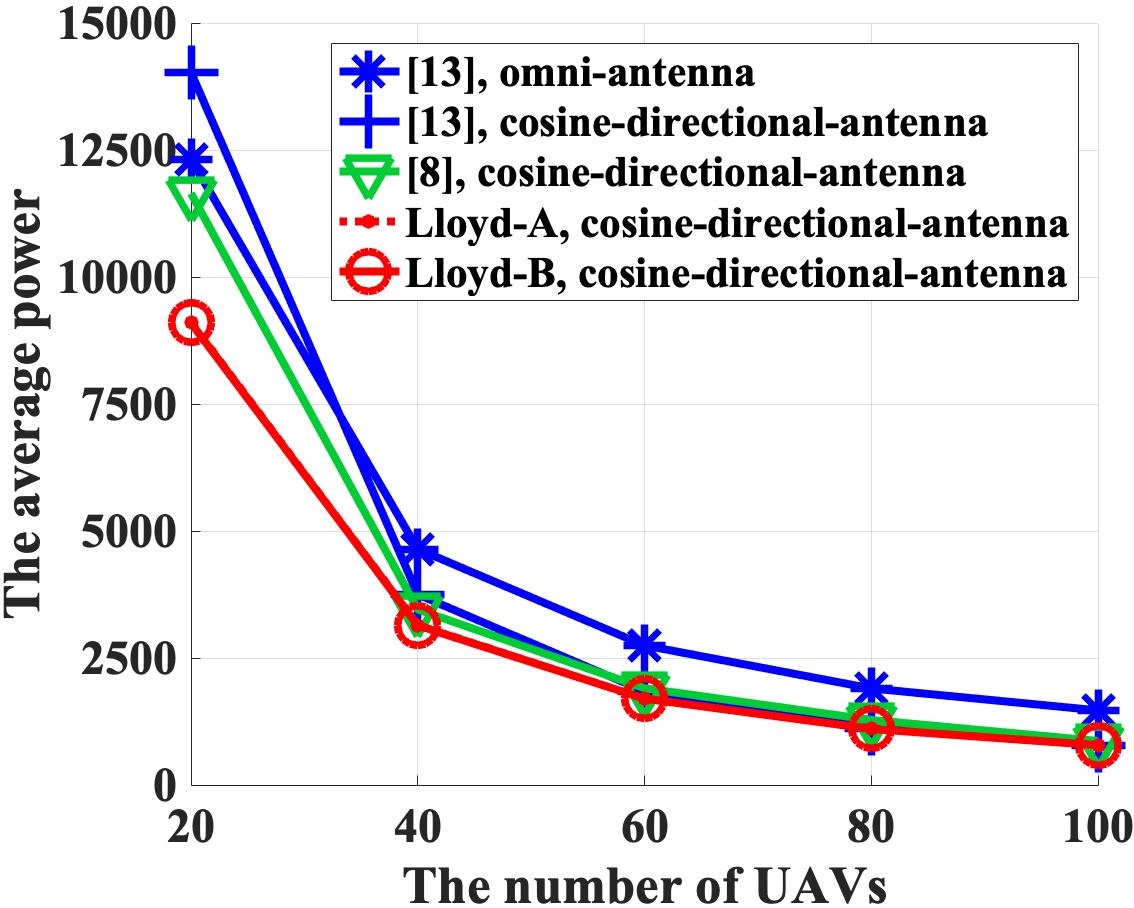}
\label{Alpha3MinH25}}
\hfil
\subfloat[]{\includegraphics[width=2in]{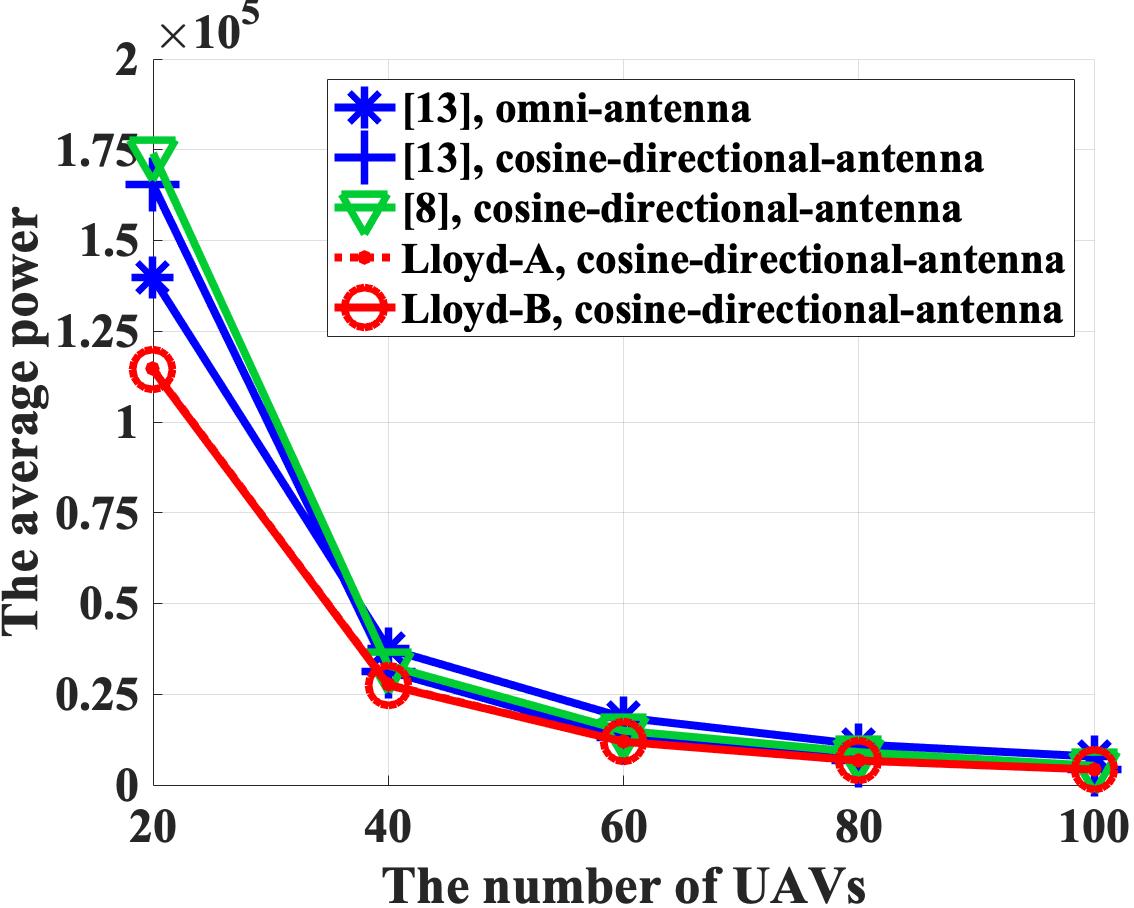}
\label{Alpha4MinH25}}
\hfil
\subfloat[]{\includegraphics[width=2in]{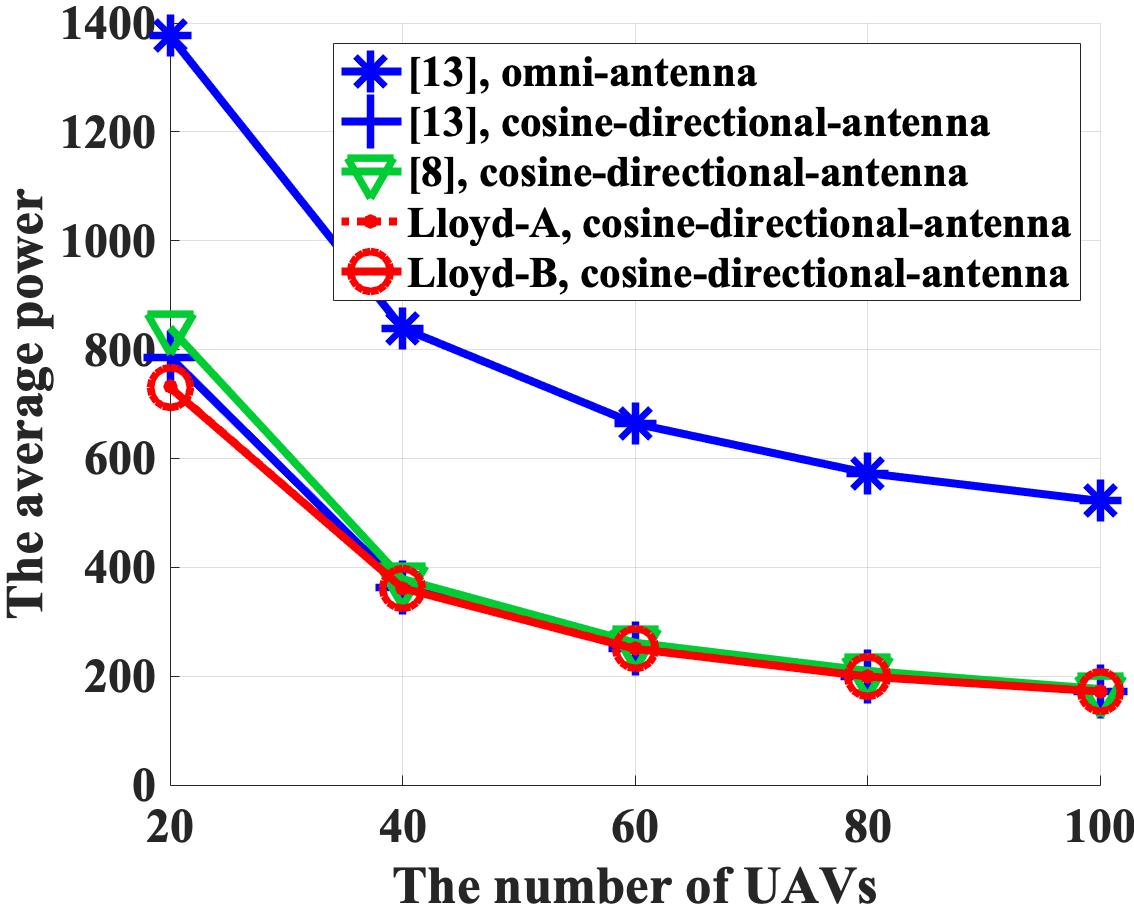}
\label{Alpha2MinH50}}
\hfil
\subfloat[]{\includegraphics[width=2in]{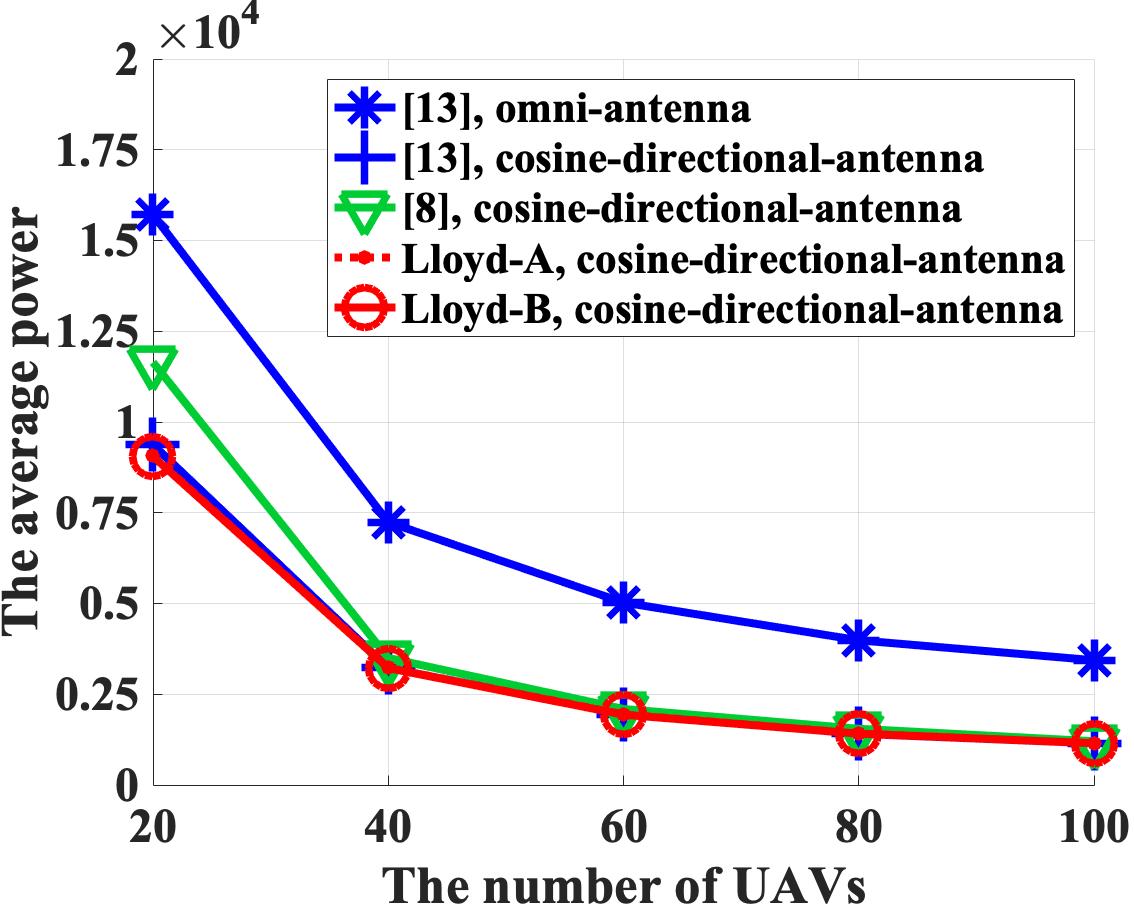}
\label{Alpha3MinH50}}
\hfil
\subfloat[]{\includegraphics[width=2in]{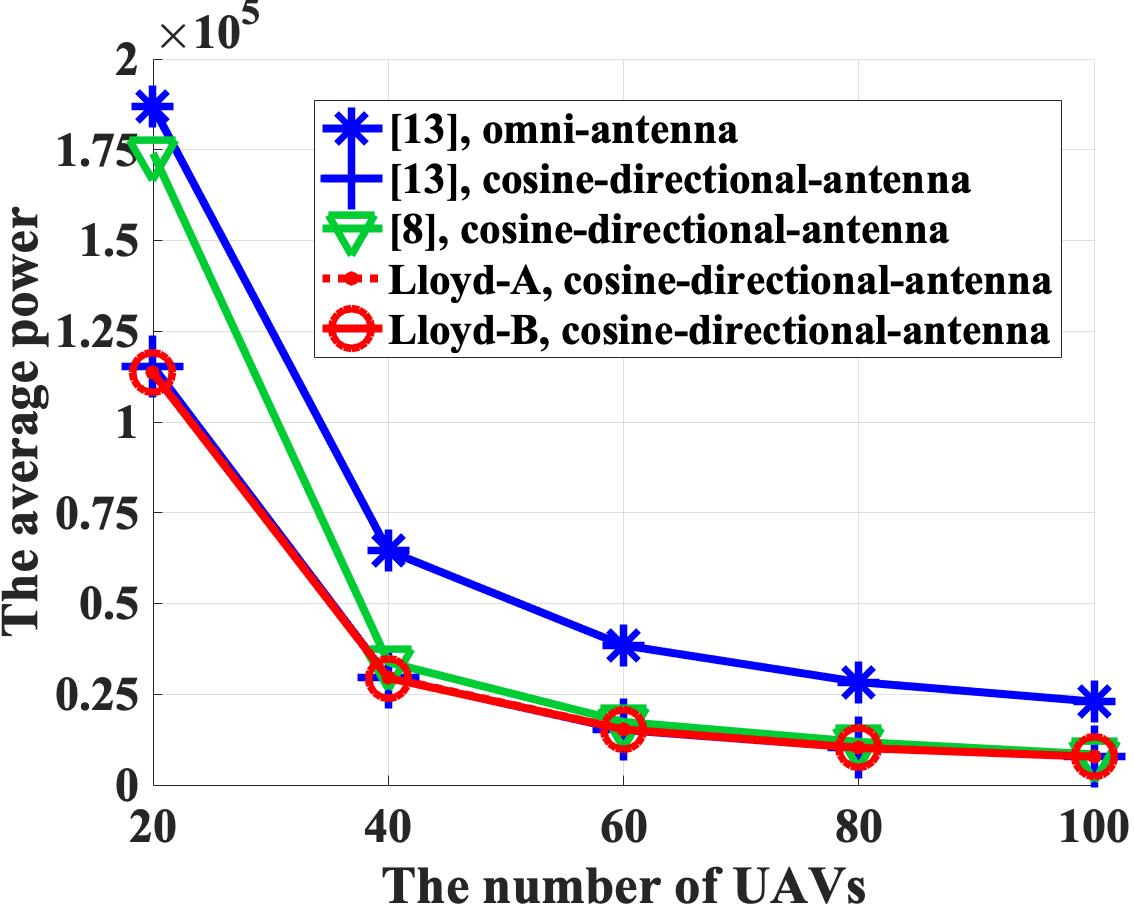}
\label{Alpha4MinH50}}
\captionsetup{justification=justified}
\caption{\small{The performance comparison for various path-loss exponents $\alp$ and various minimum flight heights with
the uniform density function  and beam-exponent $\kap=1$. (a) $\alpha=2$, $h_{min}=25$; (b) $\alpha=3$, $h_{min}=25$; (c) $\alpha=4$, $h_{min}=25$; (d) $\alpha=2$, $h_{min}=50$; (e) $\alpha=3$, $h_{min}=50$; (f) $\alpha=4$, $h_{min}=50$.}} 
\label{Distortionkap1}
\vskip1ex
\end{figure}

To evaluate the performance, we compare the average transmit-power for deployments derived by Lloyd-A, Lloyd-B, the algorithm in \cite{KKSS18}, denoted by KSS, and the algorithm in \cite{MWMM}, denoted by MSBD, with various path-loss exponents $\alp$ and various minimum flight heights\footnote{To make KSS and MSBD  Algorithms  satisfy the minimum flight height constraint, we adjust their final flight heights by $h_{min}$, i.e., $h_n = \max(h_n, h_{min})$.}.
KSS deployment Algorithm is designed to minimize the average power of UAVs with omni-directional antennas whose antenna gains are identical among all directions \pw{($\kap=0$)}. Taking a "constant" directional antenna pattern into consideration, MSBD Algorithm applies a circle packing \cite{ZGTT} to derive the ground positions of the $N$ UAVs with a common cell radius\footnote{The optimal $N-$circle packing over a square can be found, e.g., at http://hydra.nat.uni-magdeburg.de/packing/csq/csq.html.} and then determines the flight heights in terms of $\thtHPBW$.
In our simulations, the \pw{HPBW} for "constant" antenna patterns is set to $\thtHPBW=120^{\circ}$.
To make a fair comparison, the directivity parameter for cosine-shape patterns is set to $\kap=1$, which according to \eqref{eq:HPBW} corresponds to $\thtHPBW(1)=120^{\circ}$.
For $\kap=2$, we obtain $\thtHPBW(2)=90^{\circ}$.
Lloyd-like algorithms (KSS, Lloyd-A and Lloyd-B) require an initial UAV deployment.
We generate 100 initial UAV deployments randomly, i.e., every UAV location is generated according to a  uniform distribution on $1000\times1000\times100$.
Then, the Lloyd-like algorithms are initialized with the generated random deployments and the power is calculated as the average \pw{over} 100 runs.

The performance comparisons for different path-loss parameters and different minimum flight heights are shown in \figref{Distortionkap1}.
The omni-antenna power and cosine-directional-antenna power are calculated from \eqref{eq:Pbar} divided by the maximal directivity\footnote{In Section \secref{sec:model}, we assume $\bet_0(\alp)D_0(\kap)=1$ for simplicity. However, when we compare the powers of different antennas (or different $\kappa$s), the antenna directivity $D_0(\kap)$ is taken into account.} \eqref{eq:maxdirect} with $\kappa=0$ and $\kappa=1$, respectively.
Note that KSS Algorithm generates the optimal deployment for UAVs with omni-antennas.
When the number of UAVs is large, \figref{Distortionkap1} shows that KSS Algorithm can benefit from replacing omni-antennas by cosine-directional antennas. However, when the number of UAVs is small,  KSS Algorithm spends more energy on cosine-directional antennas compared to  omni-antennas.
For example, given $20$ UAVs in Fig. \ref{Alpha2MinH25}, KSS Algorithm spends the average power of $\AvPtx=1233$ on cosine-directional antennas which is larger than that of omni-antennas, i.e., $\AvPtx=1141$.
However, if $40$ UAVs are deployed, the omni-antenna power \pw{$\AvPtx=604$} exceeds the cosine-directional antenna power of  $\AvPtx=466$.

By comparing the average powers \pw{of} cosine-directional-antennas provided by different algorithms,  
we can conclude that the proposed Lloyd-B is the best solution for cosine-directional-antennas.
Moreover, cosine-directional-antennas' minimum powers which are achieved by Lloyd-B, are always smaller than omni-antennas' minimum powers which are achieved by KSS Algorithm.
In other words, cosine-directional-antennas are more energy-efficient compared to omni-antennas.
An intuitive explanation is that omni-antennas radiate power  among all directions, which is a waste of energy, while cosine-directional-antennas concentrate the power \pw{to a specific UE cell on the ground}.
Furthermore, using cosine-directional-antennas, UAVs can cover all UEs in the target area while MSBD Algorithm using a constant-directional antenna model achieves only a partial coverage, due to non-overlapping circular ground cells.
For example, in Fig. \ref{Alpha2MinH25}, MSBD deployment with constant-directional-antenna pattern only covers (or serves) 78.54\% of the target area.
In fact, the constant-directional-antenna model is an ideal but not realistic antenna model.
When the constant-directional-antenna model is replaced by the realistic  cosine-directional-antenna model, the corresponding power gain is significantly increased. Furthermore, the circle packing solution, used in MSBD Algorithm, is only available over some special-shaped target regions, e.g, squares and circles. However, Lloyd-A(B) can be applied to arbitrary target regions.

\begin{figure}[t]
\setlength\abovecaptionskip{0pt}
\setlength\belowcaptionskip{0pt}
\centering
\subfloat[]{\includegraphics[width=2in]{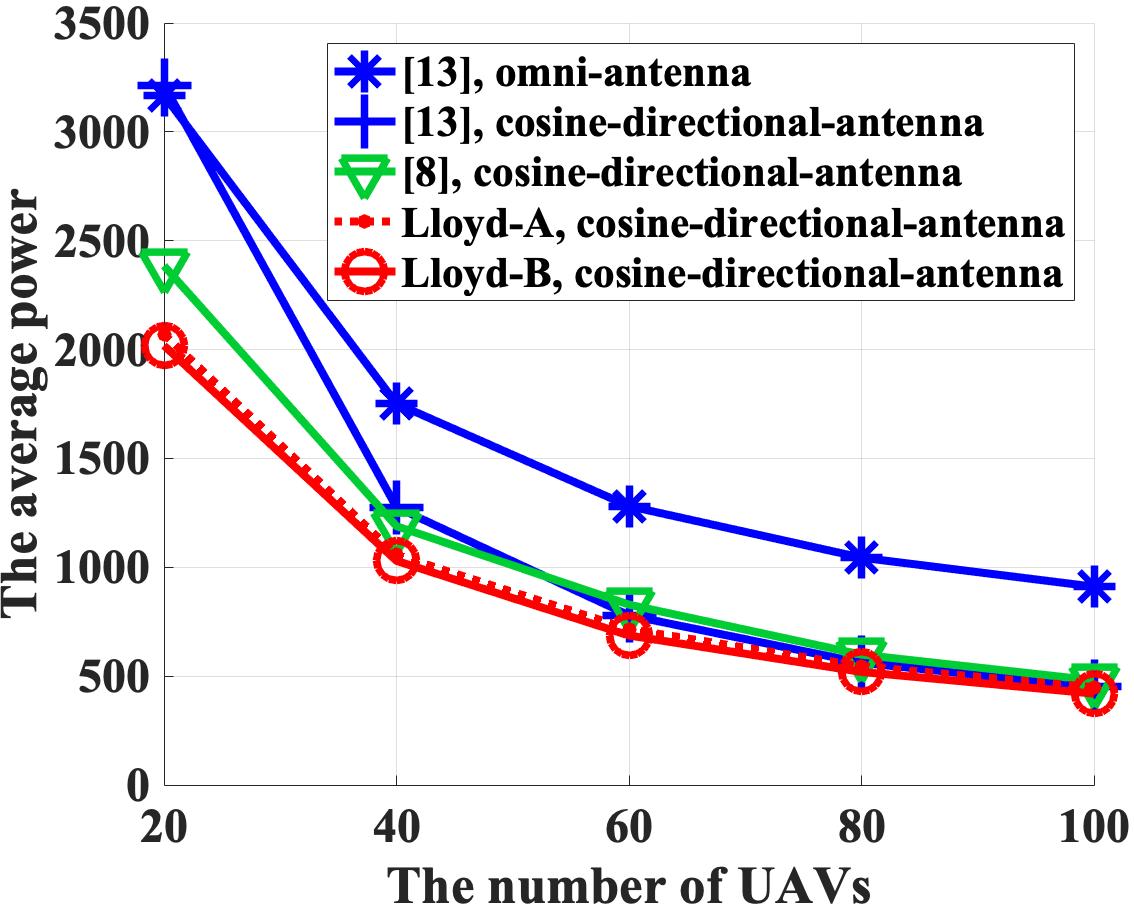}
\label{Alpha2MinH25NonUniform}}
\hfil
\subfloat[]{\includegraphics[width=2in]{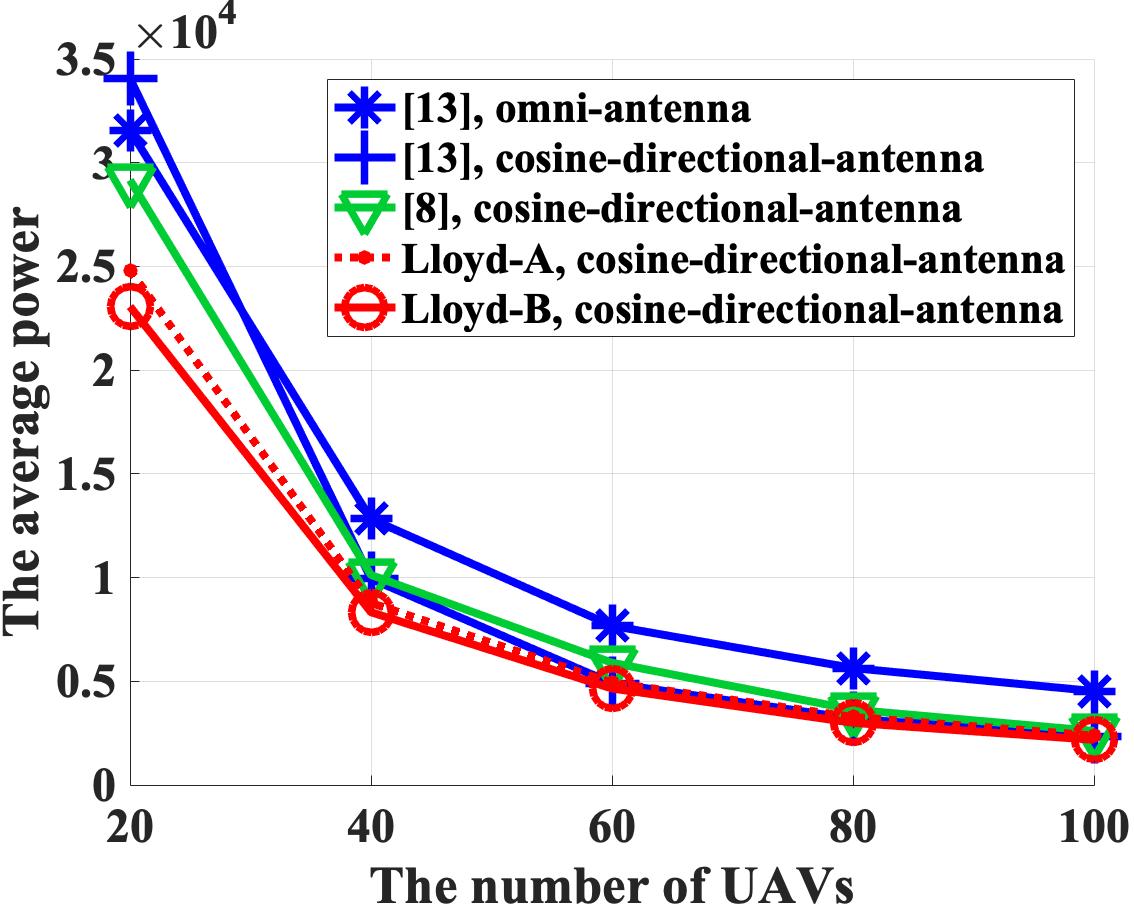}
\label{Alpha3MinH25NonUniform}}
\hfil
\subfloat[]{\includegraphics[width=2in]{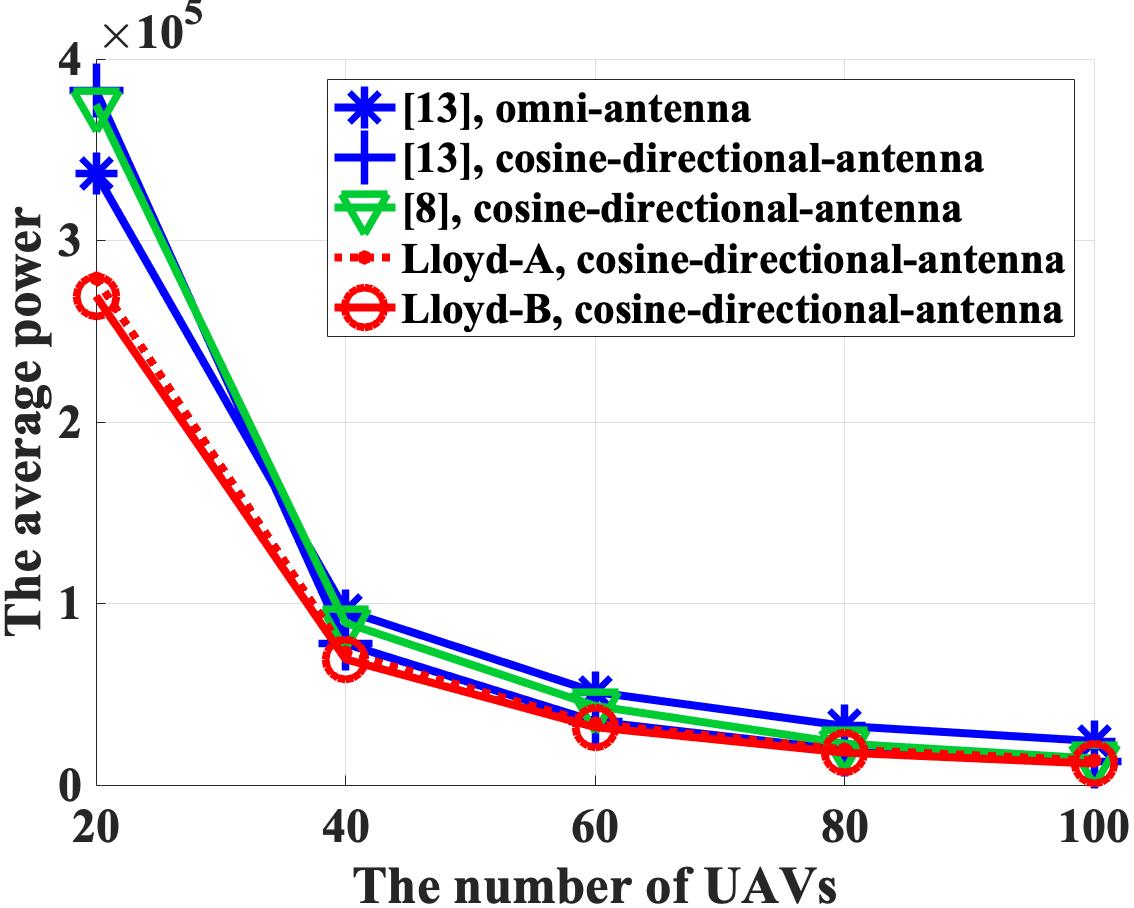}
\label{Alpha4MinH25NonUniform}}
\captionsetup{justification=justified}
\caption{\small{The performance comparison for various path-loss exponents $\alp$ and various minimum flight heights with a non-uniform density function. (a) $\alpha=2$, $h_{min}=25$; (b) $\alpha=3$, $h_{min}=25$; (c) $\alpha=4$, $h_{min}=25$.}} 
\label{Distortion}
\end{figure}

From 
\figref{Distortionkap1}, one can also find that the average powers spent by Lloyd-A and Lloyd-B are very close (the difference is less than $0.5\%$), indicating the optimality of the common height 
when the density function is uniform.
However, the gap between Lloyd-A and Lloyd-B in Figs. \ref{Alpha2MinH25NonUniform}, \ref{Alpha3MinH25NonUniform}, \ref{Alpha4MinH25NonUniform} for a non-uniform density function  is non-negligible.
For instance, given $20$ UAVs with path-loss exponent $\alpha=3$, the average power by Lloyd-A is $248$ while the average power by Lloyd-B is only $230$ which is 7.3\% lower.
As a result, the optimality of the common height cannot be extended to the scenarios with non-uniform density functions.
 Meanwhile, the minimum flight height has an influence on the deployment. If the minimum flight height is large, it may limit the capability of the algorithm to choose the best locations and force it to place the UAVs at the minimum height.
For example, Lloyd-A(B), like KSS Algorithm, places UAVs at the height of 50 in Figs. \ref{Alpha2MinH50}, \ref{Alpha3MinH50} and \ref{Alpha4MinH50}.
As a result, the average power of KSS Algorithm with cosine-directional-antennas is much closer to that of Lloyd-A(B) in Figs. \ref{Alpha2MinH50}, \ref{Alpha3MinH50} and \ref{Alpha4MinH50}.
Nonetheless, as shown in  Figs. \ref{Alpha2MinH50}, \ref{Alpha3MinH50}, and \ref{Alpha4MinH50}, Lloyd-A(B)'s power is not larger than that of KSS Algorithm even if the performance is limited by a  large minimum flight height.

\begin{figure}[t]
\centering
\subfloat[]{\includegraphics[width=2.7in]{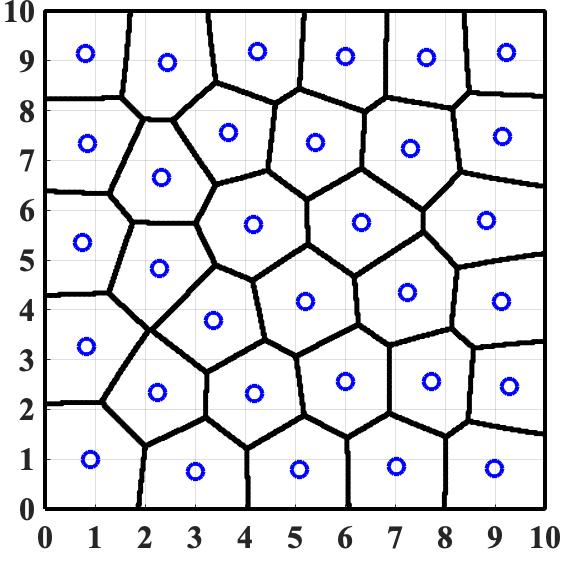}
\label{uniformPartitions32}}
\hfil
\subfloat[]{\includegraphics[width=2.7in]{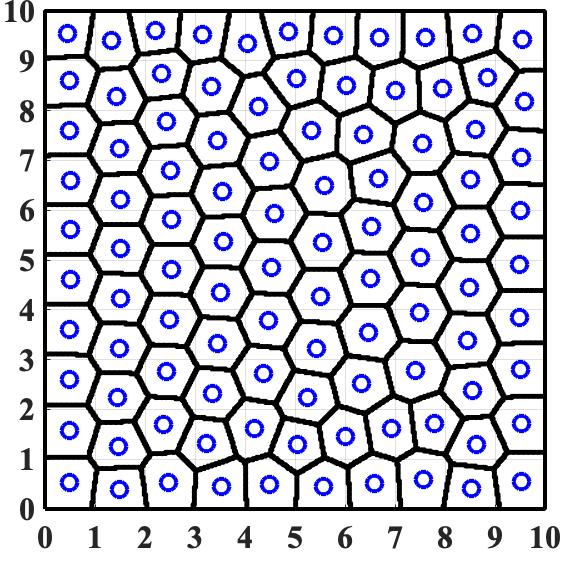}
\label{uniformPartitions100}}
\vspace{-2ex}
%
\centering
\subfloat[]{\includegraphics[width=2.6in]{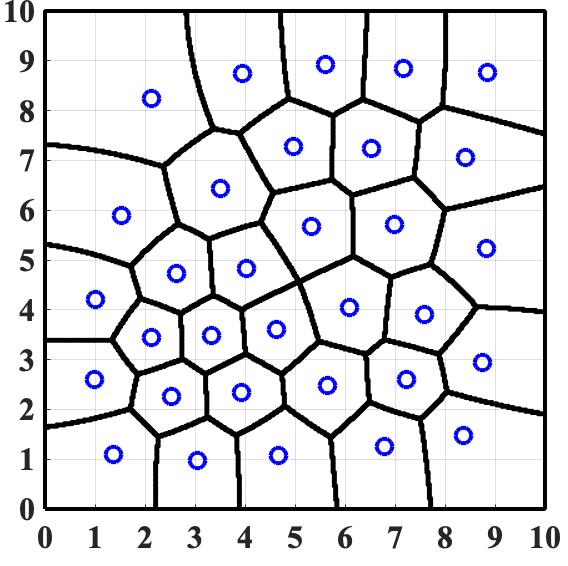}
\label{nonuniformPartitions32}}
\hfil
\subfloat[]{\includegraphics[width=2.6in]{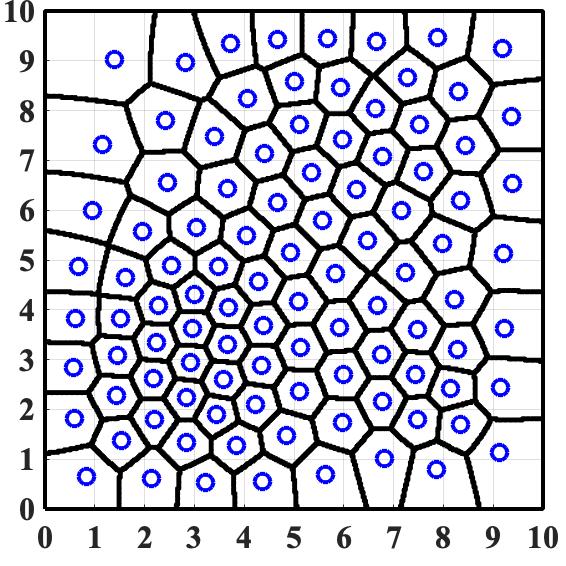}
\label{nonuniformPartitions100}}
\captionsetup{justification=justified}
\vspace{-2ex}
\caption{\small{UAV ground cells, } generalized Voronoi Diagrams, for $\alpha=2$
  and a uniform probability density function  {\protect\subref{uniformPartitions32}} with 32 UAVs, {\protect\subref{uniformPartitions100}} with 100 UAVs, and a non-uniform density  {\protect\subref{nonuniformPartitions32}} with 32
  UAVs, {\protect\subref{nonuniformPartitions100}}  with 100 UAVs.}
\label{uniformDistortionPartition3}
\end{figure}

Figs. \ref{uniformPartitions32} and \ref{uniformPartitions100} illustrate the UAV ground cells and their partitions for a uniform distribution and square region.
As the number of UAVs increases, the UAV partitions converge to
hexagons. This implies that the optimality of congruent partitioning in the one-dimensional case \cite[Thm.1]{GWJ18a}  might be valid for uniformly distributed users in the  two-dimensional case as well.
However, the UAV partitions in Figs.
\ref{nonuniformPartitions32} and \ref{nonuniformPartitions100} show that congruent partitioning is not optimal for a non-uniform distribution.

\subsection{Different antenna beamwidths}\label{sec:dab}

If we increase the antenna parameter $\kap$ in \eqref{eq:Gdirected}, the radiation gain  concentrates on a smaller
area and hence decreases the effective beamwidth. This affects the scaling factors $c(\gam,\kap)$ and directivity $D_0(\kap)$ in \thmref{thm:2D}
but not the dependence on $H$. To verify the optimal common heights derived in Theorem \ref{thm:2D}, we first
perform a brute force search to obtain the optimal height for one UAV over uniformly distributed regular hexagons with different
sizes.  For each regular hexagon, we generate $5000$ samples%
\footnote{5000 UAV heights are uniformly selected from $[0, L]$, where $L=2R$ is the length of the regular hexagon. UAV
  ground position is placed at the geometric centroid of the hexagon according to Theorem \ref{thm:2D}. }
and compare the powers.  The optimal UAV height is the one with the minimum power among the generated
samples.
\figref{fig:2Doptimalheightk} depicts the optimal common heights of one UAV over a regular hexagon 
for $\kap=1$ and $\kap=2$. The optimal common height
increases if $\kap$ increases, but the average power does not {necessarily decrease by
increasing $\kap$, as is the case in \figref{fig:D1D2D3kappa}. 
Such a conclusion is only valid for large $N$ where $\sqrt{H}\ll h^*$ and the cells are small enough such that small beams can compensate the path-loss by antenna gains.}

Moreover, we employ Lloyd-B to get the standard deviation of the optimal heights of multiple UAVs over a uniformly
distributed $10\times10$ square. From Fig.  \ref{fig:std_k2}, we find that the standard deviation
of heights decreases as the number $N$ of UAVs increases. We also observe the same trend for other $\kappa$s by  simulations. In other words, in the asymptotic regime, the UAVs tend to have
approximately an optimal common height\footnote{Due to the boundary effect, i.e., the target region cannot be perfectly divided
by congruent hexagons, the UAV heights are not exactly the same.}.

\begin{figure}[t]
\centering
\subfloat[]{\includegraphics[width=3in]{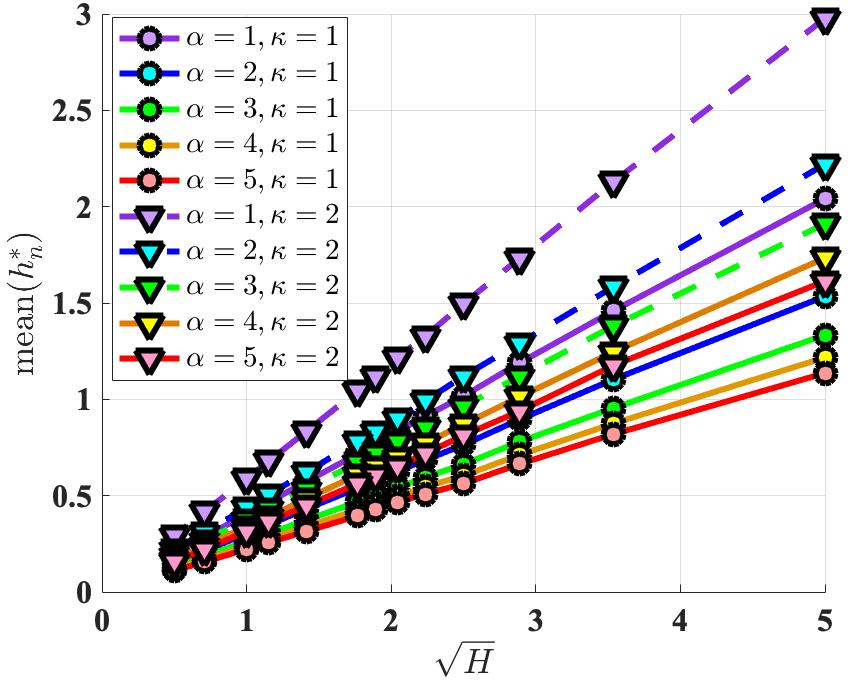}
\label{fig:k1}}
\hfil
\subfloat[]{\includegraphics[width=3in]{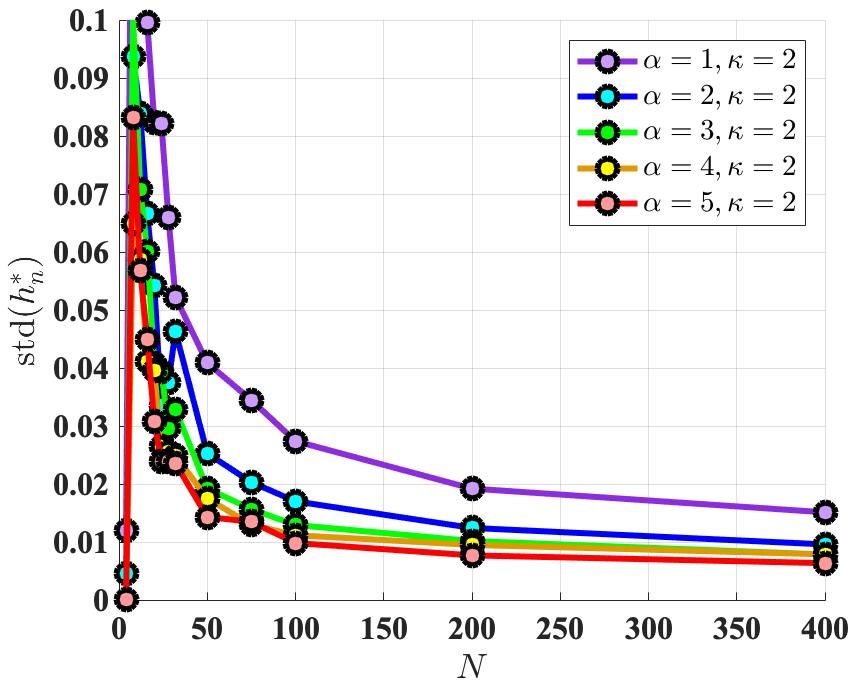}
\label{fig:std_k2}}
\captionsetup{justification=justified}
\caption{\small{{\protect\subref{fig:k1}} Optimal common height for $\alpha=1,2,3,4,5$ and $\kap=1,2$ over $\sqrt{H}=\sqrt{A/N}$, derived with brute force optimization and analytically by
\thmref{thm:2D}. {\protect\subref{fig:std_k2}} The variance (standard deviation) of all N optimal
heights converge exponentially fast in $N$ for $\kappa$ = 2.}}
\label{fig:2Doptimalheightk}
\end{figure}

\section{Conclusions}\label{sec:conclusions}
We studied a continuous coverage problem with $N$ UAVs for providing a static reliable wireless communication link to
ground users in a given target area. We adopted a realistic angle-dependent directional antenna model for the UAVs and an ideal omni-directional antenna model for the ground users. We
derived the exact average power consumption of the users to establish a reliable upload link to the UAVs at a given
bandwidth, noise power, and bit-rate.
The optimal $3$D deployment of the UAVs for minimizing the average {transmit-power} of the ground users is derived in closed form for an arbitrary path-loss exponent, antenna beamwidth, area size,
and number of UAVs.
Using the derived necessary conditions for optimal deployment, we designed Lloyd-like algorithms to minimize the transmit-power.
We demonstrated numerically
with brute-force search that asymptotically the global optimal deployment is provided by a hexagonal lattice of the UAV
ground positions and a unique common flight height.  We derived closed form solutions for the optimal common height. The optimal common height depends on the cell size per UAV, the antenna beamwidth, and the path-loss exponent.
%
{Our deployment algorithm can be used for static or airborne base-stations.}
An optimal power efficient deployment reduces interference with other wireless communications as well,
which again saves power and resources.

\appendices
\section{Proof of \lemref{lem:moebiusdia}}\label{app:moebiusregions}
 
  The minimization of the distortion functions over $\Omega$ defines an assignment rule for a
  generalized Voronoi diagram $\Vor(\gP,\bH)=\{\Vor_1,\Vor_2,\dots,\Vor_N\}$ where 
  \begin{align}
    \Vor_{n} =\Vor_n(\gP,\bH):=
      &\set{\vome\in\Ome}{ a_n\Norm{\vp_n-\vome}^2 +b_n \leq  a_m\Norm{\vp_m-\vome}^2 +b_m, m\not=n}
  \end{align}
  is the $n$th generalized Voronoi region \cite[Chap.3]{OBSC00}. Here, we denote the weights as in
  \eqref{eq:Eptx} by the positive numbers $a_n = h_n^{-\frac{\kap}{\gam}}$ and $b_n=h_n^{2-\frac{\kap}{\gam}}$
  and define a \emph{M{\"o}bius diagram} \cite{BK06b,BWY07}. The bisectors of M{\"o}bius diagrams are circles or lines
  in $\R^2$ as we will show below.
  The $n$th Voronoi region is defined by $N-1$ inequalities, which  can be written as the intersection of the $N-1$
  \emph{dominance regions} of $\vp_n$ over $\vp_m$, given by 
  \begin{align}
    \Vor_{nm}=\set{\vome\in\Ome}{ a_n\Norm{\vp_n-\vome}^2 +b_n \leq  a_m\Norm{\vp_m-\vome}^2 +b_m}.
  \end{align}
  If $h_n=h_m$, then $a_n=a_m$ and $b_n=b_m$, such that $\Vor_{nm}$ is the left half-space between $\vp_n$
  and $\vp_m$. For $a_n>a_m$, we can rewrite the inequality as 
  \begin{align*}
      \Norm{\vome}^2 -2 \skprod{\vc_{nm}}{\vome} + 
         \frac{a_n^2 \Norm{\vp_n}^2 \!+\!a_m^2 \Norm{\vp_m}^2 \!-\! a_na_m(\Norm{\vp_n}^2 \!+\!\Norm{\vp_m}^2)}{(a_n-a_m)^2} 
         + \frac{b_n-b_m}{a_n-a_m} \leq & 0,
  \end{align*}
  where the center point is given by
  \begin{align}
    \vc_{nm}=\frac{a_n\vp_n- a_m\vp_m}{a_n-a_m}=a_n \frac{\vp_n-h_{nm} \vp_m}{a_n -a_m}=\frac{\vp_n -h_{nm}\vp_m}{1-h_{nm}}, 
  \end{align}
  where we introduced the \emph{parameter ratio} of the $n$th and $m$th quantization points $h_{nm}:= a_m/a_n=\left(h_n /h_m\right)^{\frac{\kap}{\gam}}>0$.
  If $0<a_n-a_m$, which is equivalent to $h_n<h_m$, then this defines a ball (disc) and for $h_n>h_m$ its complement.
  Hence, we have
  \begin{align}
    \Vor_{nm}=\begin{cases}
      \set{\vome\in\Omega}{\Norm{\vome-\vc_{nm}}   \le r_{nm}},&  h_n<h_m\\
       \set{\vome\in\Omega}{ \Norm{\vome- \vp_n}\leq \Norm{\vome- \vp_m}}, &h_n=h_m\\
      \set{\vome\in\Omega}{\Norm{\vome-\vc_{nm}}    \ge  r_{nm}},&  h_n >h_m
    \end{cases}
  \end{align}
  where the radius square is given by
  \begin{align}
    r_{nm}^2 
            &\!\!=\! a_na_m\frac{\Norm{\vp_n-\vp_m}^2}{(a_n-a_m)^2} \!+\! \frac{b_m-b_n}{a_n-a_m}=
            \frac{a_m}{a_n}\frac{\Norm{\vp_n-\vp_m}^2}{(1-\frac{a_m}{a_n})^2}\! +\! \frac{b_m-b_n}{a_n-a_m}
            =\frac{h_{nm} \Norm{\vp_n-\vp_m}^2}{(1-h_{nm})^2} \!+\! \frac{b_m-b_n}{a_n-a_m}\notag.
  \end{align}
  The second summand can be written as
  \begin{align}
    \frac{b_m -b_n}{a_n -a_m} &
    = \frac{h_m^{2-\frac{\kap}{\gam}} - h_n^{2-\frac{\kap}{\gam}}}{h_n^{-\frac{\kap}{\gam}}-h_m^{-\frac{\kap}{\gam}}}
    = \frac{h_n^{2} \left(\left(h_n/h_m\right)^{\frac{\kap}{\gam}-2} -1\right)}{1-\left(h_n/h_m\right)^{\frac{\kap}{\gam}}} 
    = h_n^2 \frac{h_{nm}^{1-\frac{2\gam}{\kap}} -1}{1-h_{nm}}
  \label{eq:bmnanm}.
  \end{align}
  For any $\kap\geq 1,\gam\geq (1+\kap)/2$, we have $2> \kap/\gam>0$ and $1-\frac{2\gam}{\kap}<0$. Hence,  if
  $0<h_n<h_m$, then
  $h_{nm}=(h_n/h_m)^{\frac{\kap}{\gam}}< 1$ and $h_{nm}^{1-\frac{2\gam}{\kap}}>1$ and  if $h_n>h_m>0$, then $h_{nm}> 1$ and
  $h_{nm}^{1-\frac{2\gam}{\kap}}<1$.  In
  both cases \eqref{eq:bmnanm} is strictly positive, which implies a radius
  $r_{nm}>0$ even if $\vp_n=\vp_m$. In fact, the radius will only vanish if $h_n=0$, in which case the region
  will be empty. Note that a region can be empty if the dominance regions do not intersect. For $h_n=h_m$, the radius approaches infinity and the bisection is a line.
  
%

\section{Proof of \thmref{thm:2D}}\label{sec:proof2D}
For the homogeneous case with fixed common height $h=h_n$, the distortion function $d$ is given by a non-decreasing
continuous and positive function  in the Euclidean distance $r=\Norm{\vq-\vome}$ as
\begin{align}
  d(\vq,\vome) =f_{\gam,\kap}(\Norm{\vq-\vome},h)= (\Norm{\vq-\vome}^2 + h^2)^{\gam}/{h^\kap}.
\end{align}
Since we assume a uniform density, we have $\lam(\ome)=1/A$ for all $\ome\in\Ome$.

\noi The optimal deployment
problem with $h=h_n$ has centroidal ground positions given by \cite{DFG99}
\begin{align}
    \int_{\Vor^*_n} \frac{ (\Norm{\vp_n^*-\vome}^2+h^{*2})^{\gam}}{h^{*\kap}} d\vome
    = \min_{\vp\in\Clos(\Vor^*_n)}  \int_{\Vor^*_n} \frac{ (\Norm{\vp-\vome}^2+h^{*2})^{\gam}}{h^{*\kap}} d\vome,
  \label{eq:optground}
\end{align}
where $\Clos(\Vor^*_n)$ denotes the convex closure of the set $\Vor^*_n$.  Unfortunately, there is no closed form
expression for an arbitrary $\gam$.  However, asymptotically ($N\to\infty$, high-resolution) it is known that the
optimal Voronoi regions will be congruent to the regular Hexagon, i.e., $\Vor^*_n\sim\Hexa_n$ \cite{DFG99,Gru99}.
Hence, from the conditions \eqref{eq:pnopt} and \eqref{eq:hopt}, we obtain a local critical common height, if
and only if
\begin{align}
  z=h^{*2} = \frac{ \frac{1}{A} \int_{\Hexa_n} (\Norm{\vome-\vq^*_n}^2+h^{*2})^{\gam} d\vome}{\frac{2\gam}{\kap A}
  \int_{\Hexa_n}
  (\Norm{\vome-\vq^*_n}+h^{*2})^{\gam-1}d\vome}\label{eq:ncz}.
\end{align}
%
%
We know that $h^*>0$ and hence $h^{*2}=z>0$ is associated to only one height.  Since the optimal ground positions are
all centroidal and the regions $\Vor_n$ are all congruent,  asymptotically we have
\begin{align}
  \int_{\Hexa_n} (\Norm{\vome-\vq_n^*}^2+z)^{\gam} d\ome \sim \int_{\Hexa}(\Norm{\vome}^2+z)^{\gam}
  d\vome=:\MtHexa(2\gam,H).
\end{align}
Here, we centered the Hexagon $\Hexa$ such that the centroids are at the origin $\vq^*=\zero$.  The integral
$\MtHexa(2\gam,H)$ denotes a distorted polar moment of a hexagon with area $H=\mu(\Hexa)=\mu(\Vor^*_n)=A/N$.  More
precisely, the additive distortion $z$ creates a polynomial of different \emph{polar moments of the hexagon}.  Since we
 need to identify the distortion $z$ which achieves equality in \eqref{eq:ncz}, we have to calculate the polar
moments, which determine the polynomial coefficients of
%
%
\begin{align}
  g_{\gam}(z)=   \int_{\Hexa} \frac{2\gam}{\kap} z(\Norm{\vome}^2+z)^{\gam-1}d\vome
  -\int_{\Hexa}(\Norm{\vome}^2+z)^{\gam} d\vome=0.
  \label{eq:xalpall}
\end{align}
\begin{figure}
\vskip-3ex
\subfloat[]{\includegraphics[width=3.5in]{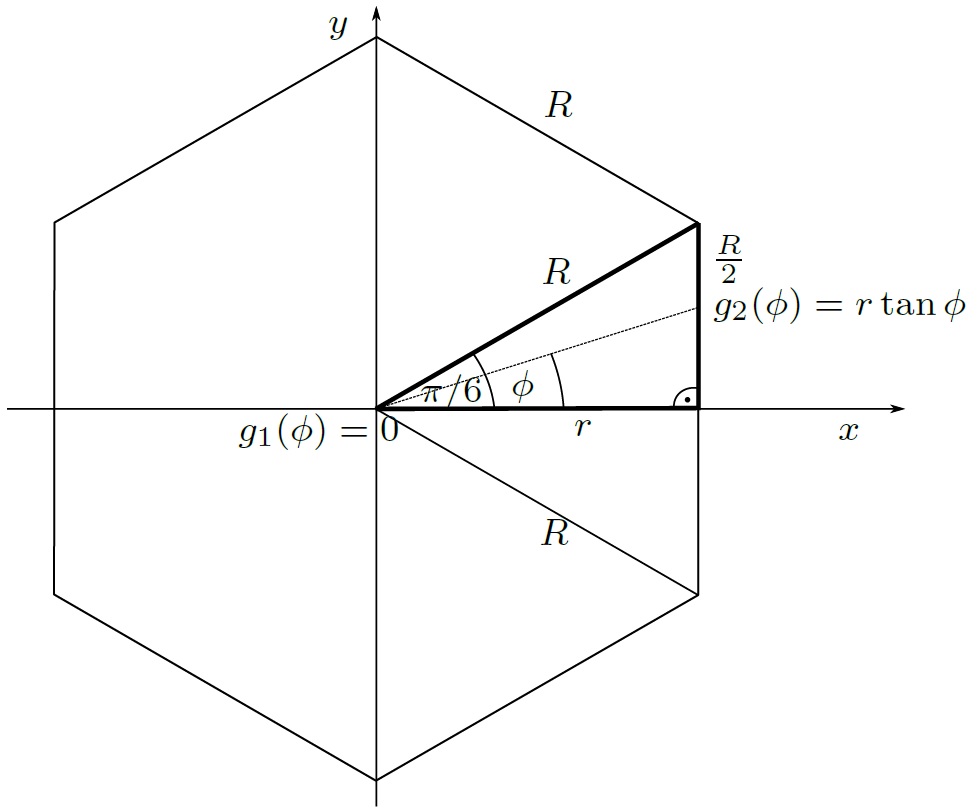}\label{fig:polar}}
\subfloat[]{\includegraphics[width=3in]{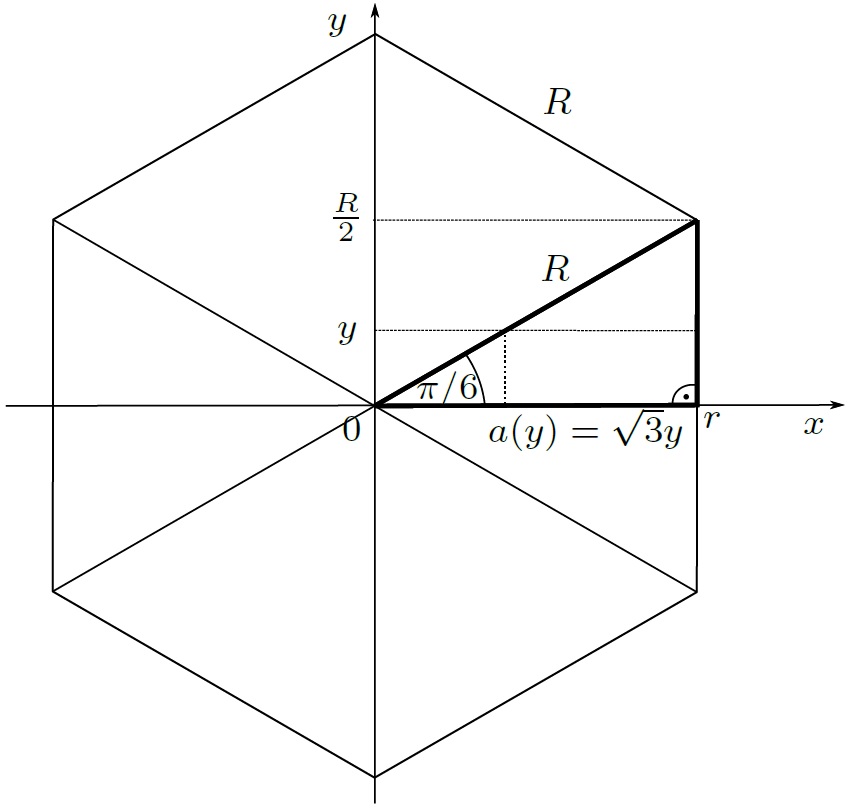}\label{fig:cart}}
\captionsetup{justification=justified}
\vspace{-2ex}
\caption{\small{Hexagon parameterization in right triangles $\Del$ with {\protect\subref{fig:polar}} polar  and
{\protect\subref{fig:cart}} Cartesian coordinates.}}
\label{fig:hexagonparam}
\vskip3ex
\end{figure}
Note that both integrals are strictly positive increasing and continuous functions in $z\geq 0$ for any real-valued
$\gam\geq1$.  Since for $z=0$ the first integral is vanishing and the second one is not, there can exist only one
$z>0$ for which the difference vanishes. Such a $z$ exists, since for $2\gamma/\kap>1$ the first integral increases
faster in $z$ compared with the second one. Hence, there exists only one optimal common height $h^*=\sqrt{z}$.  Furthermore, the
kernel integral only depends on the radius $\Norm{\vome}$ and since the hexagon $\Hexa$ consists of $12$ right-angled
triangles \Del, which are identical up to a rotation around the origin, as shown in \figref{fig:hexagonparam}, we derive the
following  equivalent condition for \eqref{eq:xalpall}:
\begin{align}
  \int_{\Del}\left( \frac{2\gam}{\kap} z(\Norm{\vome}^2+z)^{\gam-1} -(\Norm{\vome}^2+z)^{\gam} \right)d\vome=0.
\label{eq:equationDel}
\end{align}
Hence, we only have to deal with polar moments of the triangle $\Del$. To calculate solutions for specific values of
$\gam$,
we need to explicitly calculate the integrals.  The areas of the hexagon and triangles in terms of the radius
$r$ of the inscribed circle, are respectively,
\begin{align}
  H=\mu(\Hexa)= 12\int_{\Del}d\vome =   6r^2 \tan\frac{\pi}{6} = 2\sqrt{3}r^2  \quad \text{and}\quad
  \mu(\Del)=\frac{H}{12},
  \label{eq:Handt}
\end{align}
see for example \cite[4.5.3]{Zwi03} and  \figref{fig:hexagonparam}.  Note that the length of the edges are the same
as the outer radius $R=\frac{2r}{\sqrt{3}}$.
However, for arbitrary $\gam\geq1$, we need general orders of the moments
%
$    \MDelepsH=\int_{\Del} \Norm{\vome}^\eps d\vome.$
%
Since $\Norm{\vome}^\eps= f(\rho,\phi)=\rho^\eps$, is continuous in the radius $\rho$ and angle $\phi$, we
can parameterize the integral in polar coordinates, as shown in \figref{fig:polar}, to derive \cite[17.5]{Swo83}:
\begin{align}
  \MDelepsH&
  =\int_{0}^{\frac{\pi}{4}}\! \int_{g_1(\phi)}^{g_2(\phi)} \rho^{\eps}\rho d\rho
  d\phi
  =\int_0^{\frac{\pi}{4}} \!\int_0^{r\tan(\phi)} \!\rho^{\eps+1}d\rho d\phi
  =  \int_0^{\frac{\pi}{4}}\! \frac{1}{\eps+2} (r\tan\phi)^{\eps+2}d\phi.
\end{align}
Let us substitute $\tan\phi=t$, which results in $d\phi=dt/(1+\tan^2\phi)$ and hence
\begin{align}
\MDelepsH  &= \frac{1}{\eps+2}\int_{0}^1  \frac{ r^{\eps+2} t^{\eps+2}}{1+t^2} dt = \frac{r^{\eps+2}}{\eps+2} \int_0^1
  \frac{t^{\eps+2}}{1+t^2} dx.
\end{align}
By using \cite[(3.241)]{GR15}, we have
\begin{align}
 \MDelepsH &= \frac{r^{\eps+2}}{2\eps+4}\bet\left(\frac{\eps+3}{2}\right)
  = \left(\frac{H}{2\sqrt{3}}\right)^{\eps+2} \frac{1}{2\eps+4}\bet\left(\frac{\eps+3}{2}\right).
\end{align}
With the integral representation in \cite[8.375(2)]{GR15}, for odd $\eps=2n-1$, we get the expression
\begin{align}
  \bet(n+1)=(-1)^{n} \ln 2 + \sum_{i=1}^{n} \frac{(-1)^{i+n}}{i}\quad,\quad n \in \Nplus.
\end{align}
Unfortunately, a closed form expression for $\bet(n+1/2)$ or $\bet(x)$ is difficult to derive since they are given in
terms of the \emph{Gamma} or \emph{Riemann-Zeta function}. However, we can use the Cartesian parameterization to derive
the even moments for $\eps=0,2,4$.  For $\eps=0$, we get the triangle area
\begin{align}
 \MDel(0,H)= \MtHexa(0,H)/12=\mu(\Del)=\frac{H}{12}. \label{eq:Mzero}
\end{align}
The second, fourth, and sixth polar moments of the triangle is derived in \appref{sec:HexaMoments} as 
\begin{align}
  \MDel(2,H) & = \frac{5}{18\sqrt{3}}r^4 =\frac{5H^2}{216\sqrt{3}},\label{eq:Mtwo} \quad,\quad
  \MDel(4,H) = \frac{56}{270\sqrt{3}}r^6 = \frac{7}{270\cdot 9}H^3, \\
  \MDel(6,H) & = \frac{166}{35\cdot27\sqrt{3}}r^8 = \frac{83 H^4}{72\cdot 35\cdot27\sqrt{3}} \label{eq:Msix},
\end{align}
where we used \eqref{eq:Handt} for $H$. With \eqref{eq:equationDel}, this yields to the global common height for
$\gam=\kap=1$
\begin{align}
  h^*(1,H)\sim \sqrt{z}=\sqrt{\frac{\kap}{2-\kap}\frac{\MDel(2,H)}{\MDel(0,H)}}
   =c(1)\sqrt{H}, \ c(1)=\sqrt{ \frac{5 }{18 \sqrt{3}}}  \approx \sqrt{0.1603}. \label{eq:h1opt}
\end{align}
For $\gam=2$ and $3\geq \kap\geq1$, we obtain from \eqref{eq:equationDel} a quadratic equation in $z$:
\begin{align}
  g_2(z)=  &\int_{\Del} \left(\frac{4-\kap}{\kap}z^2 + \frac{4-2\kap}{\kap} z\Norm{\vome}^2 -\Norm{\vome}^4 \right) d\vome.
  \label{eq:hexagoncond}
\end{align}
Then, using \eqref{eq:Mzero} and \eqref{eq:Mtwo} in  \eqref{eq:hexagoncond}, we have
\begin{align*}
  g_2(z)  = \frac{(4\!-\!\kap)H}{12\kap} z^2 \!+\! \frac{5(2\!-\!\kap)H^2}{108\sqrt{3}\kap} z - \frac{7H^3}{270\cdot9}
  =0 \quad\LRA\quad 0 = z^2  \!+\! \frac{5H}{9\sqrt{3}}\frac{2\!-\!\kap}{4\!-\!\kap} z - \frac{14 H^2\kap}{405(4\!-\!\kap)},
\end{align*}
which has the following unique positive solution:
\begin{align}
  z&= \sqrt{ \frac{25(2-\kap)^2H^2}{18^2 3(4-\kap)^2}\! +\!\frac{14\kap H^2}{405(4-\kap)}} -\frac{5(2-\kap)H}{18\sqrt{3}(4-\kap)}
  =\frac{\sqrt{\frac{500+172\kap-43\kap^2}{5}} -5(2-\kap)}{18\sqrt{3}(4-\kap)}H
\end{align}
resulting in the optimal common height
\begin{align}
   h^*(2,\kap,H)= c(2,\kap)\sqrt{H}\text{ with }
  \  c(2,\kap)= \sqrt{\frac{\sqrt{ \frac{(172-43\kap)\kap}{5} +100}-10+5\kap}{18\sqrt{3}(4-\kap)}}.
\end{align}
Finally, for $\gam=3$ and $5\geq \kap\geq1$,  we get a cubic equation
\begin{align}
  g_3(z)= z^3 \frac{6-\kap}{\kap}\int_{\Del}
         +z^2\frac{12-3\kap}{\kap}\int_\Del \Norm{\vome}^2
         +z\frac{6-3\kap}{\kap}\int_{\Del}\Norm{\vome}^4
         - \int_{\Del}\Norm{\vome}^6.
\end{align}
Inserting the moments \eqref{eq:Mtwo}-\eqref{eq:Msix}, we obtain the coefficients $a_i$ as
\begin{align}
  0=z^3 + \frac{5H}{6\sqrt{3}}\frac{4-\kap}{6-\kap}  z^2 + \frac{14H^2}{135}\frac{2-\kap}{6-\kap} z
  - \frac{83H^3}{210\cdot 27\sqrt{3}}\frac{\kap}{6-\kap}
  =z^3+a_2 z^2+a_1 z+a_0.\label{eq:cubicz}
\end{align}
%
%
Then, one solution of \eqref{eq:cubicz} is given by
\cite[(3.8.2)]{AS64}
as
\begin{align}
  z_1= s_1^{1/3} + s_2^{1/3} -\frac{a_2}{3}  \quad\text{with}\quad s_1=p +\sqrt{q^3 +p^2}, \ s_2=p-\sqrt{q^3+p^2},
\end{align}
where we have
\begin{align}
  q &=  \frac{1}{3}a_1 -\frac{1}{9}a_2^2 = \frac{43\kap^2 -344\kap +16}{4860(6-\kap)^2}H^2,\\
  p &= \frac{1}{6}a_1 a_2 -\frac{1}{2}a_0 -\frac{1}{27}a_2^3 
  =\frac{-143360 + 16728\kap + 444 \kap^2 - 37 \kap^3}{612360\sqrt{3}(6-\kap)^3}H^3.
\end{align}
Note that the discriminant $q^3+p^2>0$ for $1\leq \kap\leq 5$ and every $H>0$. Therefore, there exists only one real-valued
solution, given by $z_1$ (all third roots are real-valued, $s_2$ can be also negative).
Then, asymptotically, the optimal height for $H$ and $\gam=3$ is
\begin{align}
  h^*(3,\kappa,H)&\sim\sqrt{z_1}
  =\sqrt{\frac{5}{18\sqrt{3}} \frac{ (u(\kap)-v(\kap))^{\frac{1}{3}} + (u(\kap)+v(\kap))^{\frac{1}{3}} -(4-k) }
  {6-\kap}H},
\end{align}
where
\begin{align}
  u(\kap) & = (143360 - 16728\kap -444 \kap^2 + 37\kap^3)/(125\cdot 35), \label{eq:uv}\\
  v(\kap) & =\frac{12(6\!-\!\kap)}{125\cdot 35}\sqrt{\frac{3}{5}}\sqrt{ 6607552\!+\! 659680 \kap \!+\! 103387 \kap^2\!
  -\! 108408\kap^3 \!+\! 9034\kap^4}.
\end{align}

Asymptotically, the minimal average distortion is given by \eqref{eq:phoptlocal} as
\begin{align}
  \AvDis^*(\gam,H)
   \sim \frac{N}{A} \int_{\Hexa} \frac{(\Norm{\vome}^2+(h^*(\gam,H))^2)^{\gam}}{h^*(\gam,H)}d\vome
  =\frac{12}{H} \int_{\Del} \frac{(\Norm{\vome}^2+(h^*(\gam,H))^2)^{\gam}}{h^*(\gam,H)}d\vome.
\end{align}
For $\gam=1$, we get with \eqref{eq:h1opt}, \eqref{eq:Mzero},  and \eqref{eq:Mtwo}:
\begin{align}
  \AvDis^* (1,H) &\sim \frac{12}{H}
 \left[ \sqrt{\frac{18\sqrt{3}}{5H} } \int_{\Del}\Norm{\vome}^2 d\vome+\sqrt{\frac{5H}{18\sqrt{3}}} \int_{\Del} d\vome \right]
 = \sqrt{10H/(9\sqrt{3})}\approx 0.8 H^{\frac{1}{2}}.
\end{align}
For $\gam=2$, we have
%
%
%
\begin{align}
   \AvDis^* (2,\kap,H) &\sim \frac{12}{H}
   \left[ \frac{1}{h^*(2,\kap,H)} \int_\Del\Norm{\vome}^4 + 2h^*(2,\kap,H) \int_{\Del} \Norm{\vome}^2 + (h^*(2,\kap,H))^3
   \frac{H}{12}\right]\notag.
\intertext{Using \eqref{eq:Mtwo}, after some simple calculations, we have}
   \AvDis^* (2,\kap,H) & \sim \left(\frac{14}{405c(2,\kap)}
+ \frac{5c(2,\kap)}{9\sqrt{3}}
+ c^3(2,\kap)\right) H^{\frac{3}{2}}\notag.
\end{align}
And, finally, for $\gam=3$, we calculate
\begin{align*}
   \AvDis^* (3,H) &\sim \frac{12}{H}
   \left[ \frac{1}{h^*(3,H)} \int_\Del\Norm{\vome}^6\! +\! 3h^*(3,H) \int_{\Del} \Norm{\vome}^4\! +\! 3(h^*(3,H))^3 \int_{\Del}
   \Norm{\vome}^2\! + \! (h^*(3,H))^5 \frac{H}{12}\right]
\intertext{using \eqref{eq:Mtwo}, \eqref{eq:Msix}, and \eqref{eq:uv}, we get for
   $\tc=\sqrt{(u-v)^{1/3}+(u+v)^{1/3}-5\sqrt{3}}$}
    \AvDis^* (3,H) &\sim  \left( \frac{83}{195\cdot 27 c(3,\kap)} + \frac{14c(3,\kap)}{135} +
    \frac{5c^3(3,\kap)}{9\sqrt{3}} + c^5(3,\kap)\right) H^{5/2}.
\end{align*}

\section{Moments of inertia over right-angled Triangles}\label{sec:HexaMoments}

A hexagon  is symmetric around the origin and can be separated in  $6$ equiangular triangles with edge length $R$. 
Hence each triangle can be split in two right triangles $\Del$ with hypotenuses  $R$ and cathetus $a=R/2=r\sqrt{3}$ and
$r$, as shown in \figref{fig:cart}. Moments of order $2n$ are then given by
\begin{align}
  \MDel(2n,H)=  \int_{\Del} \Norm{\vome}^{2n}_2 d\vome = \int_{\Del} (x^2+y^2)^n dx dy.
\end{align}
We rotate the right triangle such that its longer cathetus $r$ lies on the positive $x$ axis. Integrating $y$ from $0$ to
$a=R/2=r/\sqrt{3}$, we need to adjust the lower integral bound for $x$  by the triangle as $\tan(\pi/6)=y/a(y)$
and hence by $a(y)=\sqrt{3}y$. The second moment is
\begin{align}
  \MDel(2)&= \int_{0}^{\frac{r}{\sqrt{3}}} \int_{a(y)}^{r} (x^2 +y^2) dx dy
= \int_{0}^{\frac{r}{\sqrt{3}}} \left[\frac{1}{3} x^3+y^2 x \right]_{\sqrt{3}y}^{r}  dy
=\frac{5r^4}{18\sqrt{3}}.
\end{align}
Similarly, for the fourth  moment, $n=2$, using $a(y)=\sqrt{3}y$ and after some calculations, we have 
\begin{align}
  \MDel(4) 
  &= \int_{0}^{r/\sqrt{3}} \frac{1}{5}(r^5-9\sqrt{3}y^5)+y^4 (r-\sqrt{3}y) +\frac{2}{3}y^2(r^3-3\sqrt{3}y^3)  dy,\notag\\
 \MDel(4) 
 &=  \frac{5r^6}{30\sqrt{3}} +\frac{r^6}{30\cdot 9 \sqrt{3}}
    +\frac{r^6}{27\sqrt{3}} =  \frac{46r^6}{10\cdot 27\sqrt{3}}    +\frac{10r^6}{10\cdot 27\sqrt{3}}
    =  \frac{28}{5\cdot 27\sqrt{3}}r^6. 
\end{align}
And, finally, for $n=3$, the sixth moment can be calculated as 
\begin{align}
  \MDel(6) &=  \int_{0}^{\frac{r}{\sqrt{3}}} \!\int_{\sqrt{3}y}^{r} (x^6+3x^4y^2 +3x^2y^4 + y^6) dx dy
  = \frac{166}{35\cdot27\sqrt{3}}r^8. 
\end{align}

\appendices


\begin{thebibliography}{1}
\bibitem{GWJ19}
J. Guo, P. Walk, and H. Jafarkhani, ``Quantizers with parameterized distortion measures,§ in Data Compression Conf.
(DCC), Mar. 2019.

\bibitem{ZYS11}
P. Zhan, K. Yu, and A. L. Swindlehurst, ``Wireless relay communications with unmanned aerial vehicles: Performance
and optimization,§ IEEE Transactions on Aerospace and Electronic Systems, vol. 47, no. 3, pp. 2068每2085, Jul. 2011.

\bibitem{ZWZ19}
Y. Zeng, Q. Wu, and R. Zhang, ``Accessing from the sky: A tutorial on uav communications for 5g and beyond,§ Mar.
2019. arXiv: 1903.05289.

\bibitem{BJL}
B. Galkin, J. Kibilda, and L. A. DaSilva, ``Backhaul for low-altitude UAVs in urban environments,§ in ICC, May 2018.

\bibitem{MSF}
M. M. Azari, F. Rosas, and S. Pollin, ``Reshaping cellular networks for the sky: Major factors and feasibility,§ in 2018
IEEE International Conference on Communications (ICC), Oct. 2018. arXiv: 1710.11404v2.

\bibitem{HA}
H. Shakhatreh and A. Khreishah, ``Maximizing indoor wireless coverage using UAVs equipped with directional antennas,§
May 2017. arXiv: 1705.09772v1.

\bibitem{HSYR}
H. He, S. Zhang, Y. Zeng, and R. Zhang, ``Joint altitude and beamwidth optimization for UAV-enabled multiuser
communications,§ IEEE Commun. Lett., vol. 22, no. 2, Feb. 2018.

\bibitem{MWMM}
M. Mozaffari, W. Saad, M. Bennis, and M. Debbah, ``Efficient deployment of multiple unmanned aerial vehicles for
optimal wireless coverage,§ IEEE Commun. Lett., vol. 20, no. 8, pp. 1647每1650, Aug. 2016.

\bibitem{YPWS-B19}
Z. Yang, C. Pan, K. Wang, and M. Shikh-Bahaei, ``Energy efficient resource allocation in UAV-enabled mobile edge
computing networks,§ 2019. arXiv: 1902.03158.

\bibitem{AE-KLY17}
M. Alzenad, A. El-Keyi, F. Lagum, and H. Yanikomeroglu, ``3D placement of an unmanned aerial vehicle base station
(UAV-BS) for energy-efficient maximal coverage,§ IEEE Trans. Wireless Commun., vol. 6, no. 4, pp. 434每437, Aug.
2017.

\bibitem{AKL14}
A. Al-Hourani, S. Kandeepan, and S. Lardner, ``Optimal LAP altitude for maximum coverage,§ IEEE Wireless Communications
Letters, vol. 3, no. 6, pp. 569每572, Dec. 2014.

\bibitem{GJcom18}
J. Guo, E. Koyuncu, and H. Jafarkhani, ``A source coding perspective on node deployment in two-tier networks,§ IEEE
Trans. Commun., vol. 66, no. 7, pp. 3035每3049, Jul. 2018.

\bibitem{KKSS18}
E. Koyuncu, R. Khodabakhsh, N. Surya, and H. Seferoglu, ``Deployment and trajectory optimization for UAVs: A
quantization theory approach,§ in 2018 IEEE (WCNC), Apr. 2018. arXiv: 1708.08832v5.

\bibitem{MLW15}
P. Jankowski-Mihulowicz, W. Lichon, and M. Weglarski, ``Numerical model of directional radiation pattern based on
primary antenna parameters,§ Int. J. of Electronics and Telecommunications, vol. 61, no. 2, pp. 191每197, Jul. 2015.

\bibitem{Bal05a}
C. A. Balanis, Antenna Theory: Analysis and Design, 3rd ed. Wiley-Interscience, 2005.

\bibitem{Erdem16}
E. Koyuncu and H. Jafarkhani, ``On the minimum distortion of quantizers with heterogeneous reproduction points,§ Data
Compression Conference, Mar. 2016.

\bibitem{KJ17}
〞〞, ``On the minimum average distortion of quantizers with index-dependent distortion measures,§ IEEE Transactions
on Signal Processing, vol. 65, no. 17, pp. 4655每4669, Sep. 2017.

\bibitem{GHHF08}
A. Gusrialdi, S. Hirche, T. Hatanaka, and M. Fujita, ``Voronoi based coverage control with anisotropic sensors,§ Jun.
2008.

\bibitem{ML}
M. Moarref and L. Rodrigues, ``An optimal control approach to decentralized energy-efficient coverage problems,§ 3,
vol. 47, Elsevier BV, Aug. 2014, pp. 6038每6043.

\bibitem{BWY07}
J.-D. Boissonnat, C. Wormser, and M. Yvinec, ``Curved voronoi diagrams,§ in Effective Computational Geometry for
Curves and Surfaces. Springer, 2007.

\bibitem{GJ}
J. Guo and H. Jafarkhani, ``Sensor deployment with limited communication range in homogeneous and heterogeneous
wireless sensor networks,§ IEEE Trans. Wireless Commun., vol. 15, no. 10, pp. 6771每6784, Oct. 2016.

\bibitem{GJ18}
〞〞, ``Movement-efficient sensor deployment in wireless sensor networks,§ IEEE Trans. Wireless Commun., vol. 18,
pp. 3469每3484, Jul. 2019.

\bibitem{MLCS}
M. T. Nguyen, L. Rodrigues, C. S. Maniu, and S. Olaru, ``Discretized optimal control approach for dynamic multi-agent
decentralized coverage,§ in IEEE International Symposium on Intelligent Control (ISIC), Sep. 2016.

\bibitem{SJH}
S. Karimi-Bidhendi, J. Guo, and H. Jafarkhani, ``Using quantization to deploy heterogeneous nodes in two-tier wireless
sensor networks,§ Jul. 2019. IEEE ISIT, arxiv: 1901.06742.

\bibitem{OBSC00}
A. Okabe, B. Boots, K. Sugihara, and S. N. Chiu, Spatial Tessellations: Concepts and Applications of Voronoi Diagrams,
2nd ed. John Wiley \& Sons, 2000.

\bibitem{KMR}
K. Venugopal, M. C. Valenti, and R. W. Heath, ``Device-to-device millimeter wave communications: Interference,
coverage, rate, and finite topologies,§ IEEE Trans. Wireless Commun., vol. 15, no. 9, pp. 6175每6188, Sep. 2016.

\bibitem{Gol05}
A. Goldsmith, Wireless Communications. Cambridge University Press, 2005.

\bibitem{ITU09}
Guidelines for evaluation of radio interface technologies for IMT-advanced, ITU recommendation M.2135-1, Geneva,
Switzerland: ITU, 2009.

\bibitem{AG18}
A. Al-Hourani and K. Gomez, ``Modeling cellular-to-UAV path-loss for suburban environments,§ IEEE Wireless Communications
Letters, vol. 7, no. 1, pp. 82每85, Feb. 2018.

\bibitem{GR15}
I. S. Gradshteyn and I. M. Ryzhik, Table of Integrals, Series, and Products, 8th ed., A. Jeffrey and D. Zwillinger, Eds.
Academic Press, 2015.

\bibitem{MSBD16a}
M. Mozaffari, W. Saad, M. Bennis, and M. Debbah, ``Unmanned aerial vehicle with underlaid device-to-device communications:
Performance and tradeoffs,§ IEEE Trans. Wireless Commun., vol. 15, no. 6, pp. 3949每3963, Jun. 2016.

\bibitem{ZXZ19}
Y. Zeng, J. Xu, and R. Zhang, ``Energy minimization for wireless communication with rotary-wing UAV,§ IEEE Transactions
on Wireless Communications, vol. 18, no. 4, pp. 2329每2345, Apr. 2019.

\bibitem{CMB05}
J. Cort{\'{e}}s, S. Mart赤nez, and F. Bullo, ``Spatially-distributed coverage optimization and control with limited-range interactions,§
ESAIM, vol. 11, no. 4, pp. 691每719, Sep. 2005.

\bibitem{GWJ18a}
J. Guo, P. Walk, and H. Jafarkhani, ``Quantizers with parameterized distortion measures,§ Nov. 2018. arXiv: 1811.02554.

\bibitem{DFG99}
Q. Du, V. Faber, and M. Gunzburger, ``Centroidal voronoi tessellations: Applications and algorithms,§ SIAM Review,
vol. 41, no. 4, pp. 637每676, Oct. 1999.

\bibitem{ZGTT}
Z. Gaspar and T. Taranai, ``Upper bound of density for packing of equal circles in special domains in the plane,§ Periodica
Polytechnica Civil Engineering, vol. 44, no. 1, pp. 13每32, Apr. 2000.

\bibitem{BK06b}
J.-D. Boissonnat and M. I. Karavelas, ``On the combinatorial complexity of euclidean voronoi cells and convex hulls of
d-dimensional spheres,§ INRIA, Jul. 2002.

\bibitem{Gru99}
P. M. Gruber, ``A short analytic proof of {F}ejes {T}{\'{o}}th's theorem on sums of moments,§ Aequ. math., vol. 58, no. 3,
pp. 291每295, Nov. 1999.

\bibitem{Zwi03}
D. Zwillinger, Standard Mathematical Tables and Formulae, 31st ed. CRC, 2003.

\bibitem{Swo83}
E. W. Swokowski, Calculus with Analytic Geometry. Prindle, Weber \& Schmidt, 1983.

\bibitem{AS64}
M. Abramowitz and I. A. Stegun, Handbook of Mathematical Functions with Formulas, Graphs, and Mathematical
Tables. Dover Publications Inc., 1964.
\end{thebibliography}
\end{document}